\documentclass[11pt,onecolumn,draftcls]{IEEEtran}

\usepackage{graphicx}
\usepackage{amsthm,amsfonts,amssymb,amsmath}
\usepackage{algorithm}
\usepackage{algpseudocode}
\usepackage{color}
\usepackage{bm,bbm}
\newtheorem{prop}{Proposition}

\newtheorem{thm}{Theorem}
\newtheorem{rmk}{Remark}
\newtheorem{lemma}{Lemma}
\newtheorem{defi}{Definition}
\newtheorem{exe}{Example}
\newtheorem{cor}{Corollary}

\DeclareMathOperator{\SNR}{\mathsf{SNR}}
\providecommand{\abs}[1]{\ensuremath{\left\lvert #1 \right\rvert}}

\newcommand{\C}{\mathbb{C}}

\newcommand{\vol}{\mbox{vol}}
\newcommand{\xx}{\mathbf{x}}
\newcommand{\yy}{\mathbf{y}}
\newcommand{\zz}{\mathbf{w}}
\newcommand{\uu}{\mathbf{u}}
\newcommand{\ww}{\mathbf{w}}
\newcommand{\HH}{\mathbf{H}}
\newcommand{\EE}{\mathbf{E}}
\newcommand{\UU}{\mathbf{U}}
\newcommand{\BB}{\mathbf{B}}
\newcommand{\YY}{\mathbf{Y}}
\newcommand{\ZZ}{\mathbf{W}}
\newcommand{\XX}{\mathbf{X}}
\newcommand{\FF}{\mathbf{F}}
\newcommand{\RR}{\mathbf{R}}
\newcommand{\nozero}{\backslash\left\{0\right\}}

\begin{document}

\title{Universal Lattice Codes for MIMO Channels}

\author{Antonio Campello, Cong Ling and Jean-Claude Belfiore\thanks{

This work was supported in part by FP7 project PHYLAWS (EU FP7-ICT 317562). Partial results of this work were presented at the International Symposium on Information Theory, 2016 \cite{OurISIT} and at the Information Theory Workshop, 2016 \cite{OurITW} .

A. Campello is currently with the Department of Electrical and Electronic Engineering, Imperial College London (e mail: a.campello@imperial.ac.uk). His work was partially done in the Department of Communications and Electronics, Telecom ParisTech, France, funded by FAPESP under grant 2014/20602-8. 

C.  Ling  is  with  the  Department  of  Electrical  and  Electronic  Engineering, Imperial College London, London SW7 2AZ, U.K. (e-mail: cling@ieee.org). 

J.-C.   Belfiore    is    with the  Mathematical and Algorithmic Sciences Lab, France Research Center,
Huawei Technologies   (e-mail: belfiore@telecom-paristech.fr).

}}

\maketitle
\begin{abstract}
We propose a coding scheme that achieves the capacity of the compound MIMO channel with algebraic lattices. Our lattice construction exploits the multiplicative structure of number fields and their group of units to absorb ill-conditioned channel realizations. To shape the constellation, a discrete Gaussian distribution over the lattice points is applied. These techniques, along with algebraic properties of the proposed lattices, are then used to construct a sub-optimal de-coupled coding schemes that achieves a gap to compound capacity by decoding in a lattice that does not depend of the channel realization. The gap is characterized in terms of algebraic invariants of the codes, and shown to be significantly smaller than previous schemes in the literature. We also exhibit alternative algebraic constructions that achieve the capacity of ergodic fading channels.
\end{abstract}
\section{Introduction}
We consider a MIMO channel with $n$ receive antennas and $m$ transmit antennas, described by the equation
\begin{equation}
\mathbf{y} = \mathbf{H}\mathbf{x}+\mathbf{w},
\label{eq:oneUse}
\end{equation}
where $\mathbf{H}\in \mathbb{C}^{n\times m}$ is the channel matrix, and $\mathbf{x} \in \mathbb{C}^m$ is the input subject to the power constraint $E[\mathbf{x}^{\dagger}\mathbf{x}]\leq m P$. The noise entries of $\mathbf{w}$ are circularly symmetric complex Gaussian with zero-mean and variance $\sigma_w^2$. We assume that the receiver has complete knowledge of $\mathbf{H}$, which is fixed during a whole transmission block. Consider the set $\mathbb{H}$ of all channel matrices with fixed (white-input) capacity $C$:
\begin{equation}\begin{split}
\mathbb{H} = \left\{\mathbf{H} \in \mathbb{C}^{n\times m}: \right. \log \det\left(\mathbf{I}+\SNR\mathbf{H}^{\dagger}\mathbf{H}\right) =\left.C\right\}.
\end{split}
\label{eq:channelSpace}
\end{equation}
This can be viewed as a compound channel with capacity $C$. The compound channel model (\ref{eq:channelSpace}) arises in several important scenarios in communications, such as the outage formulation in the open-loop mode and broadcast \cite{Ordentlich15}.

We say that a sequence of codes is \textit{universal} or \textit{achieves the capacity of the compound model} for the MIMO channel if, for all $H \in \mathbb{H}$ the error probability vanishes, as the blocklength $T \to \infty$, with rate $R$ arbitrarily close to $C$. In this work we construct universal algebraic lattice codes for the MIMO channel.

\subsection{Discussion and Organization of the Work}
Initial research on lattice codes for fading channels was concerned with the diversity order and minimum product distance \cite{ViterboOggier}. Recently, \cite{Ordentlich15} and \cite{DBLP:journals/corr/LuzziV15} have built universal codes that achieve a constant gap to the capacity in the MIMO channels. The work \cite{HeshamElGamal04} showed the existence of lattice codes achieving the optimal diversity-multiplexing tradeoff of MIMO channels. Further, \cite{YonaFeder14,Punekar15} examined the diversity order of lattice codes, in the infinite-constellation setting, for MIMO and block-fading channels, respectively. The Poltyrev limit and dispersion on ergodic fading channels were studied in \cite{DBLP:journals/corr/Vituri13}. 

The notion of compound MIMO channels dates back at least to \cite{RootVarayia1968}. The authors provide a technique to convert traditional random codes into universal ones, under the  assumption that the norm of $\mathbf{H}$ is bounded (see also \cite{ShiWesel07}). However the methods used are unstructured and do not provide any insight on the development of more practical universal codes.

In this paper, we make a step towards this goal by proving that lattice codes from generalized versions of construction A achieve the capacity of the compound MIMO channel over the entire space of channels \eqref{eq:channelSpace}. This represents an advantage of ideal lattices over the classic Gaussian random codes \cite{RootVarayia1968,ShiWesel07} and standard Construction A \cite{HeshamElGamal04}. This is made possible by exploiting the multiplicative structure of number fields and their group of units. Similar techniques had previously demonstrated good simulation performance in the fast fading channel with efficient decoding \cite{GhayaViterboJC} and optimal asymptotic diversity-versus-multiplexing tradeoff for $2 \times 2$ MIMO channels \cite{OthmanLuzziBelfiore10}.


Our contributions are listed as follows.
\begin{itemize}
\item We show that lattices constructed via algebraic number theory universally achieve the capacity of the MIMO channel. The construction is divided in two steps: first we define good infinite constellations for the MIMO channel and then we show how to shape the constellation with the lattice Gaussian distribution. Our approach shows that constellations built form number-theoretic tools can achieve not only asymptotic parameters (such as the DMT), but also the capacity of the compound channel.

\item In \cite{HeshamElGamal04}, it is shown that linear filtering equalization (multiplication by the MMSE-GDFE matrix), followed by lattice decoding in an ``equalized'' lattice achieves the optimal DMT of MIMO channels. Through the lattice Gaussian distribution \cite{LingBelfiore}, we provide an interpretation for the MMSE-GDFE matrix: if the sent point is sampled from a lattice Gaussian distribution, then MMSE-GDEF followed by lattice decoding is \textit{equivalent to MAP decoding}.
\item In Section \ref{sec:decoupled} we provide a more efficient sub-optimal scheme that achieves the compound capacity up to a constant gap. In this scheme the decoder first handles the fading matrix $\mathbf{H}$ and then performs lattice decoding in the coding lattice itself, independently of $\mathbf{H}$. This notion of efficiency follows  \cite{Ordentlich15}, where the authors consider integer-forcing achieving a gap to capacity in the compound MIMO channel. Besides reducing the gap of \cite{Ordentlich15}, we provide a characterization of the gap to capacity in terms of algebraic parameters. For instance, it is shown that algebras/number-fields whose unit-lattice have small volume minimize the gap.
\item In Section \ref{sec:Ergodic} we show how an adaptation of the previous methods can be used to achieve the capacity of the ergodic fading channel. Leveraging from algebraic techniques, our construction improve the two previous proposed lattice codes: It improves on the probability of error of \cite{DBLP:journals/corr/Vituri13} and completely eliminates the gap to capacity of \cite{DBLP:journals/corr/LuzziV15} (note, however, that our scheme currently requires statistical knowledge of the channel, which is also the case of \cite{DBLP:journals/corr/Vituri13} but not of \cite{DBLP:journals/corr/LuzziV15}).
\end{itemize}
A technical novelty of the present work is the error probability analysis of lattice Gaussian distribution via properties of \textit{sub-Gaussian} random variables. This greatly simplifies the analysis of standard lattice Gaussian codes \cite{LingBelfiore} and provides achievable results under weaker assumptions on the channel.

As a final remark, the authors of \cite{DBLP:journals/corr/LuzziV15} pose the existence of an analogue of the Minkowski-Hlawka theorem suitable for fading channels as an open problem. The results in sections \ref{sec:consA} provide such an analogue for block-fading channels.

\section{Notation and Initial Definitions}
\label{sec:model}
The channel equation \eqref{eq:oneUse} after $T$ uses can be written in matrix form:
\begin{equation}\label{eq:block-fading}
\underbrace{\mathbf{Y}}_{n\times T}=\underbrace{\mathbf{H}}_{n\times m}\underbrace{\mathbf{X}}_{m \times T}+\underbrace{\mathbf{W}}_{n\times T}
\end{equation}
where $T$ is the coherence time (codeword length). Vectorizing this equation, we obtain
\begin{equation}\label{eq:block-fading-vectorized}
\underbrace{\mathbf{y}}_{\times 1}=\underbrace{\mathcal{H}}_{nT\times mT}\underbrace{\mathbf{x}}_{mT \times 1}+\underbrace{\mathbf{w}}_{nT\times 1}
\end{equation}
where $\mathcal{H}=\mathbf{I}_T\otimes \mathbf{H}$.
We denote the Frobenius norm of $\HH$ by $\left\| \mathbf{H} \right\| = \sqrt{\text{trace}(\mathbf{H}^{\dagger}\mathbf{H})}$. The pseudo-inverse of $\HH$ will be denoted by $\HH^* = (\HH^\dagger \HH)^{-1} \HH^\dagger$.

\subsection{Complex Lattices}
A (complex) lattice $\Lambda$ is a discrete additive subgroup of $\mathbb{C}^m$. We will only consider full rank lattices, i.e., when $\Lambda$ is not contained in any proper subspace of $\mathbb{C}^m$. In this case, $\mathbb{C}$ is a free abelian group of rank $2m$ and there exists a full rank  matrix $\mathbf{B}_c \in \mathbb{C}^{m\times 2m}$ such that
\begin{equation}
\Lambda = \mathcal{L}(\mathbf{B}_c) = \left\{ \mathbf{B}_c \xx: \xx \in \mathbb{Z}^{2m} \right\}.
\label{eq:complexLattice}
\end{equation}
A complex lattice has an equivalent real lattice generated by the matrix obtained by stacking real and imaginary parts of matrix $\mathbf{B}_c$:
$$\mathbf{B}_r = \left(\begin{array}{c} \Re(\mathbf{B}_c) \\ \Im(\mathbf{B}_c)\end{array} \right) \in \mathbb{R}^{2m \times 2m}.$$
\begin{exe}
The simplest example of complex lattices are $\mathbb{Z}[i]$-lattices, where $\mathbb{Z}[i] = \left\{ a + bi: a, b \in \mathbb{Z}\right\}$ is the set of Gaussian integers. A $\mathbb{Z}[i]$-lattice has the form
$$\Lambda = \left\{ \mathbf{B} \xx: \xx \in \mathbb{Z}[i]^{m} \right\},$$
 where $\mathbf{B} = \Re(\mathbf{B})+i \Im(\mathbf{B})\in \mathbb{C}^{m\times m}$. In the notation of \eqref{eq:complexLattice}, $\Lambda$ is generated as a free abelian group by a matrix $\mathbf{B}_c$ whose first $m$ columns are $\Re(\mathbf{B})+i \Im(\mathbf{B})$ and last $m$ columns are  $-\Im(\mathbf{B})+i \Re(\mathbf{B})$. Its equivalent real lattice has generator matrix 

$$\overline{\mathbf{B}}=\left(\begin{array}{cc}  \Re(\mathbf{B}) &  -\Im(\mathbf{B}) \\  \Im(\mathbf{B}) &  \Re(\mathbf{B})\end{array}\right).$$
 \end{exe}
 In general, operations with complex lattices can be done by operating their real equivalent. We define the dual of a complex lattice as;
 \begin{equation}
 \Lambda^* = \left\{ \yy \in \mathbb{C}^m : \Re(\yy^\dagger \xx)  \in \mathbb{Z}, \,\, \forall \xx \in \Lambda \right\}.
 \end{equation}
 
Identifying $\mathbb{C}^{m}$ with $\mathbb{R}^{2m}$ through the mapping $\psi(\xx) = \left( \Re(\xx),\Im(\xx) \right)$ this is an extension the notion of dual to the complex space. In particular, the real equivalent of $\Lambda^*$ coincides with the dual of $\psi(\Lambda)$.

The volume of a complex lattice $\Lambda$ is denoted by $V(\Lambda)$ and defined as the volume of its equivalent real lattice, i.e. $V(\mathcal{L}(\BB_c)) = |\det \BB_r|$ . For a $\mathbb{Z}[i]$-lattice, $V(\Lambda) = |\det \mathbf{B}|^2$.
The \textit{Voronoi region} of a point $\xx\in \Lambda$ is defined as 
$$\mathcal{V}_{\Lambda}(\xx) \triangleq \left\{ \mathbf{y} \in \mathbb{C}^n : \left\| \xx - \mathbf{y} \right\| \leq \left\| \bar{\xx}-\mathbf{y}\right\| \mbox{ for all } \bar{\xx} \in \Lambda \right\}.$$ 
Throughout the text,  we write $\mathcal{V}_{\Lambda}=\mathcal{V}_{\Lambda}(\textbf{0})$. The volume of $\Lambda$ is equal to the volume of its Voronoi region, viewed as a region in $\mathbb{R}^{2n}$. Given $\sigma > 0$, the \textit{volume-to-noise ratio} (VNR) of a lattice is defined as $\gamma_{\Lambda}(\sigma) = V(\Lambda)^{1/n}/\sigma^2$.

For applications in coding for the MIMO channel it is useful to represent the vectors of $\Lambda$ in matrix form; this can be done in a straightforward way. If $\Lambda \subset \mathbb{C}^{mT}$ is a full-rank lattice, the matrix form representation of a point $\xx=(x_1,\ldots,x_{mT}) \in \Lambda$ is
$$\mathbf{X}= \left( \begin{array}{cccc} x_1 & x_2 & \cdots & x_T \\ x_{T+1} & x_{T+2} & \cdots & x_{2T} \\
 x_{2T+1} & x_{2T+2} & \cdots & x_{3T} \\ \vdots & \vdots & \ddots & \vdots \\  x_{(m-1)T+1} & x_{(m-1)T+2} & \cdots & x_{mT}\end{array} \right).$$

\subsection{The Lattice Gaussian Distribution}

For~$\sigma>0$ and $\mathbf{c} \in \C^m$,
the continuous Gaussian distribution of covariance matrix $\Sigma$ centered at ${\bf c}$ is given by
\begin{equation*}
 f_{\sqrt{\Sigma},{\bf c}}(\mathbf{x})=\frac{1}{{\pi}^m\det(\Sigma)}e^{- (\mathbf{x}-{\bf c})^{\dagger}\Sigma^{-1}(\mathbf{x}-{\bf c})},
\end{equation*}
for $\mathbf{x} \in\C^m$. For convenience, we write $f_{\sqrt{\Sigma}}(\mathbf{x})=f_{\sqrt{\Sigma},{\bf
0}}(\mathbf{x})$. Consider the $\Lambda$-periodic function 
\begin{equation}\label{Guass-function-lattice}
  f_{\sqrt{\Sigma},\Lambda}(\mathbf{x})=\sum_{{{\bm \lambda}}\in \Lambda}
{f_{\sqrt{\Sigma},{{\bm \lambda}}}(\mathbf{x})}=\frac{1}{{\pi}^m\det(\Sigma)}
\sum_{{\bm \lambda} \in \Lambda} e^{- (\mathbf{x}-{\bf \lambda})^{\dagger}\Sigma^{-1}(\mathbf{x}-{\bf \lambda})},
\end{equation}
for all $\mathbf{x} \in\C^m$. Observe that $f_{\sigma,\Lambda}$ restricted to a fundamental region $\mathcal{R}(\Lambda)$ is a probability density. We define the \emph{discrete Gaussian distribution} over $\Lambda$ centered at $\mathbf{c} \in \C^n$ as the following discrete distribution taking values in ${\bm \lambda} \in \Lambda$:
\[
D_{\Lambda,\sqrt{\Sigma},\mathbf{c}}({\bm \lambda})=\frac{f_{\sqrt{\Sigma},\mathbf{c}}(\mathbf{{\bm \lambda}})}{f_{\sqrt{\Sigma},\mathbf{c}}(\Lambda)}, \quad \forall {\bm \lambda} \in \Lambda,
\]
where $f_{\sqrt{\Sigma},\mathbf{c}}(\Lambda) \triangleq \sum_{{\bm \lambda} \in
\Lambda} f_{\sqrt{\Sigma},\mathbf{c}}(\mathbf{{\bm \lambda}})=f_{\sqrt{\Sigma},\Lambda}(\mathbf{c})$. Again for convenience, we write $D_{\Lambda,\sqrt{\Sigma}}=D_{\Lambda,\sqrt{\Sigma},\mathbf{0}}$. 


The flatness factor of a lattice~$\Lambda$ quantifies the maximum variation of~$f_{\sqrt{\Sigma},\Lambda}(\mathbf{x})$ for~$\mathbf{x} \in \C^m$.

\begin{defi} [Flatness factor]
For a lattice~$\Lambda$ and for covariance matrix~$\sqrt{\Sigma}$, the flatness factor
is defined by:
\begin{equation*}
\epsilon_{\Lambda}(\sqrt{\Sigma})  \triangleq \max_{\mathbf{x} \in
\mathcal{R}(\Lambda)}\abs{
V(\Lambda)f_{\sqrt{\Sigma},\Lambda}(\mathbf{x})-1}.
\end{equation*}
\end{defi}

In words, $\frac{f_{\sqrt{\Sigma},\Lambda}(\mathbf{x})}{1/V(\Lambda)}$, the ratio between $f_{\sqrt{\Sigma},\Lambda}(\mathbf{x})$ and the uniform distribution over~$\mathcal{R}(\Lambda)$, is within the range $[1-\epsilon_{\Lambda}(\sqrt{\Sigma}), 1+\epsilon_{\Lambda}(\sqrt{\Sigma})]$.

\begin{prop} [Expression of $\epsilon_{\Lambda}(\sqrt{\Sigma})$] \label{expression_epsilon}
We have:
\begin{eqnarray*}
  \epsilon_{\Lambda}(\sqrt{\Sigma}) &=& \frac{V(\Lambda)}{{\pi}^m\det(\Sigma)}
\sum_{{\bm \lambda} \in \Lambda} e^{- {\bm \lambda}^{\dagger}\Sigma^{-1}{\bm \lambda}} \\
   &=& \sum_{{\bf \bm{\lambda}^*}\in \Lambda^*}e^{-\pi^2 {\bm \lambda}^{\dagger}\Sigma^{-1}{\bm \lambda}}-1
\end{eqnarray*}
In particular, if $\Sigma=\sigma^2\mathbf{I}$, then
\begin{eqnarray*}
  \epsilon_{\Lambda}({\sigma}) &=& \left(\frac{\gamma_{\Lambda}(\sigma)}{{\pi}}\right)^{m}{
\Theta_{\Lambda}\left({\frac{1}{\pi\sigma^2}}\right)}-1 \\
   &=& \Theta_{\Lambda^*}\left({{\pi\sigma^2}}\right)-1
\end{eqnarray*}
where $\gamma_{\Lambda}(\sigma) = \frac{
V(\Lambda)^{1/m}}{\sigma^2}$ is the volume-to-noise ratio (VNR), and $\Theta_{\Lambda}(\tau)=\sum_{\bm{\lambda} \in \Lambda} e^{-\pi\tau\|\bm{\lambda}\|^2}$ is the theta series.
\end{prop}

%

The significance of a small flatness factor is two-fold. Firstly, it assures the ``folded" distribution $f_{\sqrt{\Sigma},\Lambda}(\mathbf{x})$ is flat; secondly, it implies the discrete Gaussian distribution $D_{\Lambda,\sqrt{\Sigma},\mathbf{c}}$ is ``smooth". We refer the reader to \cite{LLBS_12,LB_13} for more details.

The following lemma is a generalization of Regev's and is particularly useful for communications and security \cite{LuzziLV16}.

\begin{lemma}\label{lem-Regev-Generalized}
Let $\mathbf{x}_1$ be sampled from discrete Gaussian distribution $D_{\Lambda+\mathbf{c},\sqrt{\Sigma_1}}$ and $\mathbf{x}_2$ sampled from continuous Gaussian distribution $f_{\sqrt{\Sigma_2}}$. Let $\Sigma_0 = \Sigma_1 + \Sigma_2$ and let $\Sigma_3^{-1} = \Sigma_1^{-1} +\Sigma_2^{-1}$. If $\epsilon_{\Lambda}(\sqrt{\Sigma_3}) \leq \varepsilon \leq \frac{1}{2}$, then the distribution $g$ of $\mathbf{x}=\mathbf{x}_1+\mathbf{x}_2$ is close to $f_{\sqrt{\Sigma_0}}$:
\[
g(\mathbf{x}) \in f_{\sqrt{\Sigma_0}}(\mathbf{x})\left[ {1-4\varepsilon}, 1+4\varepsilon \right].
\]
\end{lemma}

This lemma has considerable implications. It implies, for instance, that the discrete
Gaussian distribution over a lattice is a capacity-achieving input distribution if
the flatness factor tends to zero \cite{LB_13}. 
\subsection{The Minkoswki-Hlawka Theorem}
A crucial result to prove the achievability of lattice coding schemes is the Minkoswki-Hlawka Theorem. 
Let $\psi(\xx) = \left( \Re(\xx),\Im(\xx) \right)$ be the mapping that identifies $\mathbb{C}^{m}$ with $\mathbb{R}^{2m}$. The following is an adaptation of the classical Minkowski-Hlawka theorem (see e.g. \cite[Ch. 7]{ZamirBook}).
\begin{thm} Let $m \geq 1$ be fixed and $f:\mathbb{R}^{2m} \to \mathbb{R}$ be an integrable function that vanishes outside a bounded support. For any $\varepsilon > 0$, there exists a random ensemble of full-rank lattices $\mathbb{L}_m = \left\{\Lambda \right\} \subset \mathbb{C}^{m}$ and volume $V$ such that 
\begin{equation} E\left[ \sum_{x\in\Lambda\backslash\{\mathbf{0}\}} f(\psi(\xx)) \right] \leq V^{-1} \int_{\mathbb{R}^{2m}} f(\xx) \text{d}\xx + \varepsilon,
\end{equation}
where the expectation is taken with respect to some measure in $\mathbb{L}_n$.
\label{thm:basicMH}
\end{thm}

For real lattices, Loeliger \cite{Loeliger} proved that a possible random ensemble $\mathbb{L}_n$ satisfying Theorem \ref{thm:basicMH} can be constructed from error-correcting codes using the so-called Construction A. Ling et. al \cite{LLBS_12} generalized this theorem for certain functions whose support is not bounded; this is applicable, for instance, to calculate the average behavior of the flatness factor.

\section{The Infinite Compound Channel}
\label{sec:InfiniteModel}
\subsection{Infinite Compound Model}
Since our coding schemes is divided in two parts, shaping and coding, we first define a compound model for the infinite lattice constellation, analogous to the Poltyrev limit \cite{Loeliger} for Gaussian channels. In this model with unconstrained power, we are interested in finding the minimum VNR ratio for which it is possible to communicate with vanishing probability of error.

Let 
\begin{equation}
\mathbb{H}_{\infty} = \left\{\HH \in \mathbb{C}^{n\times m}: |\det \HH^{\dagger}\HH| = D \right\},
\label{eq:compoundModel}
\end{equation}
where $D$ is a positive constant. Consider a  lattice $\Lambda \subset \mathbb{C}^{m T}$. The error probability of a lattice scheme $\Lambda$, given $\textbf{H}$, is denoted by $P_e(\Lambda,\mathbf{H})$.

\begin{defi} We say that a sequence of lattices $\Lambda_T$ of increasing dimension $mT$ is universally good for the MIMO channel if for any VNR $\gamma_{\Lambda_T}(\sigma) > \frac{\pi e}{D^{1/m}}$ and all $\mathbf{H} \in \mathbb{H}_\infty$, $ P_e(\Lambda_T,\mathbf{H}) \to 0$.
\label{def:goodLattices}
\end{defi}
Notice that the condition on the VNR is equivalent to $\gamma_{(\mathbf{I}_T \otimes\mathbf{H})\Lambda_T}(\sigma) > \pi e$.
We stress that this definition requires a sequence of lattices to be simultaneously good for \textit{all} channels in the set. For a fixed $\mathbf{H}$, this requirement is not different from the original Gaussian channel coding problem. However, as shown in the end of this section traditional codes \cite{Loeliger} fail to achieve the infinite compound capacity of $\mathbb{H}_{\infty}$ under lattice decoding.

Another way of interpreting Definition \eqref{def:goodLattices} is that a universally good sequence of lattices achieves vanishing probability of error for any channel realization with \textit{normalized-log-density}
\begin{equation}\delta(\Lambda_T) = \frac{1}{mT}\log V(\Lambda_T) \to \log\left(\frac{\pi e}{D^{1/m}}\right) \mbox{ as } T \to \infty
\end{equation}

\subsection{General Results}
Suppose that $\HH \in \mathbb{H}_{\infty}$ (Eq. \eqref{eq:compoundModel}), and let $\tilde{\mathbb{H}}_\infty = \mathbb{H}_\infty/D^{1/2m}$ be the normalized ensemble of channel matrices. To achieve the infinite compound capacity, we first show how to ``compactify'' $\mathbb{H}_\infty$. 

\begin{defi} Let $\mathbb{L}$ be an ensemble of matrix-form lattices in dimension $m \times T$. We say that $\mathbb{L}$ compacifies $\tilde{\mathbb{H}}_{\infty}$ if for any $\tilde{\HH} \in \tilde{\mathbb{H}}$ there exists matrices ${E}_{\tilde{\HH}} \in \mathbb{C}^{n\times m}$, $U_{\tilde{\HH}} \in  \mathbb{C}^{m\times m}$ such that
\begin{enumerate}
\item[(i)] $\left\| \EE_{\tilde{\HH}}\right\| \leq \alpha$ for a universal constant $\alpha$ not depending of $\tilde{\HH}$.
\item[(ii)] $\mathbb{L}$ is invariant under multiplication by $U_{\tilde{\HH}}$.
\end{enumerate}
\label{def:compacify}
\end{defi}

Compactification handles ill-conditioned channel realizations, by bounding the norm of the ``error matrix'' $\EE_{\tilde{\HH}}$. We have the following result.

\begin{thm} Suppose that $\mathbb{L}_T \subset{C}^{mT}$ is a sequence of Minkowski-Hlawka ensembles of lattices with volume $V>0$  that compactifies $\tilde{\mathbb{H}}_{\infty}$. There exists a sequence of lattices $\Lambda_T \subset \mathbb{L}_T$ universally good for the MIMO channel.
\label{thm:infiniteLimit}
\end{thm}
The proof uses the techniques of \cite{RootVarayia1968}, \cite[Appendix]{ShiWesel07}, and consists of three parts: (i) a good lattice for a \textit{fixed} $\HH$, (ii) a universal code for a finite set of channel matrices and (iii) fine quantization of the possible channel realizations. We start with the simple observation that lattice decoding in the complex channel model is equivalent to the real one, i.e., if $\Lambda \in \mathbb{C}^{mT}$ and $\yy$ is the received vector,

\begin{equation*}
\arg \min_{\xx \in \Lambda} \left\| \yy - \xx \right\| = \arg \min_{\psi(\xx) \in \psi(\Lambda)} \left\| \psi(\yy) - \psi(\xx) \right\|.
\end{equation*}

Furthermore, a circularly symmetric Gaussian distribution with variance $\sigma_w^2$ corresponds to a two-dimensional real Gaussian distribution with covariance $(\sigma_w^2/2) \mathbf{I}_2$. For convenience we set $V=1/\lambda$.
\begin{proof}
(i) For a given non-random matrix $\HH$, it was proven in \cite[Thm. 3]{HeshamElGamal04}, following the steps of \cite{Loeliger}, that the Minkowski-Hlawka theorem implies the existence of a sequence of lattices which are good for the MIMO channel $\YY = \HH \XX + \ZZ$. Specifically, by applying linear zero-forcing $\HH^{*} \YY = \XX + \HH^{*} \ZZ$, followed by an ``ambiguity decoder'', the probability of error goes to zero as $T\to \infty$ as long as $\text{VNR} > \pi e/{D}^{1/m}$.
Here we consider a small variation. The receiver first finds $\EE_{\tilde{\HH}}$ satisfying (i) and (ii) in Definition \ref{def:compacify}, and calculates
\begin{equation}\tilde{\YY} = D^{-1/2m} \EE_{\tilde{\HH}}^{*}\YY = \tilde{\XX} + \tilde{\mathbf{W}},
\end{equation}
where $\tilde{\ZZ} = D^{-1/2m} \EE_{\tilde{\HH}}^{*}\ZZ$ and $\tilde{\XX} = \UU_{\tilde{\HH}}\XX \in \tilde{\Lambda}.$ Due to Definition \ref{def:compacify}, $\tilde{\Lambda}$ is also in the ensemble, and averaging over all $\tilde{\Lambda}$ is the same as averaging over all $\Lambda$. From \cite[Thm. 3]{HeshamElGamal04}:
\begin{equation}\begin{split}P_e(\Lambda,\HH) &\leq  P(\psi(\mathbf{W}) \notin \mathcal{B}_{\sqrt{mT(\sigma_w^2+\varepsilon)}})+ P_e(\Lambda,\HH | \psi(\mathbf{W}) \in \mathcal{B}_{\sqrt{mT(\sigma_w^2+\varepsilon)}})\\ &\leq P(\psi(\mathbf{W}) \notin \mathcal{B}_{\sqrt{mT(\sigma_w^2+\varepsilon)}}) + (1+\delta)\lambda {D}^{-{T}}{ \mbox{vol } \mathcal{B}_{\sqrt{mT(\sigma_w^2+\varepsilon)}}},
\end{split}
\end{equation}
where $\delta, \varepsilon \to 0$ as $m \to \infty$, and the probability can be made arbitrarily small for $\text{VNR} > \pi e/{D}^{1/m}$. The balls $\mathcal{B}_{\sqrt{mT(\sigma_w^2+\varepsilon)}}$ are in $\mathbb{R}^{2mT}$ (the adaptations in the corresponding models to the complex case were made in view of the observation that follows Theorem 2).

(ii) Suppose now that we have $L$ channel matrices $\HH_1,\ldots,\HH_L$. Averaging the \textit{sum} of the probabilities $\mathcal{P}_e(\Lambda,\HH_i)$ over all lattices in the ensemble, we have
\begin{equation}
\begin{split}& \EE_\mathbb{L} \left[ \sum_{i=1}^L \mathcal{P}_e(\Lambda,\HH_i)\right] =   \sum_{i=1}^L \EE_\mathbb{L} \left[\mathcal{P}_e(\Lambda,\HH_i)\right] \\ 
&\leq L \left(P(\mathbf{W} \notin \mathcal{B}_{\sqrt{mT(\sigma_w^2+\varepsilon)}}) + (1+\delta)\lambda {D}^{-{T}}{ \mbox{vol } \mathcal{B}_{\sqrt{mT(\sigma_w^2+\varepsilon)}}}\right),
\end{split}
\label{eq:averageProb}
\end{equation}
Again, this sum of probabilities can be made arbitrarily small as long as the threshold VNR condition is satisfied. 

(iii) For the third part we need the assumption that $\mathbb{L}_T$ compactifies the channel space. For two channel realizations $\HH_0$ and $\HH$ such that the corresponding error matrices satisfy $\left\| \mathbf{E}_{\tilde{\HH}_0}- \mathbf{E}_{\tilde{\HH}}\right\| \leq \eta,$ it is proven in Appendix \ref{app:proofs} that
\begin{equation}\begin{split}\mathcal{P}_e(\Lambda, \HH | \mathbf{W} \in \mathcal{B}_{\sqrt{mT(\sigma_w^2+\varepsilon)}} ) \leq e^{\eta \alpha (mT(1+\varepsilon/\sigma_w^2))}\mathcal{P}_e(\Lambda, \HH_0 | \mathbf{W} \in \mathcal{B}_{\sqrt{mT(\sigma_w^2+\varepsilon)}} ).
\end{split}
\label{eq:boundCompact}
\end{equation}
Since the set of possible error matrices is compact, for any arbitrarily small $\eta$, we can choose $L=L_{\eta,m}$ large enough and matrices $\EE_{\tilde{\HH}_1},\ldots,\EE_{\tilde{\HH}_L}$ such that for all $\EE$, we can find $\EE_{\tilde{\HH}_i}$ satisfying $\left\| \EE_{\tilde{\HH}_i}-\EE \right\| \leq \eta$. Therefore, for any $\HH \in \mathbb{H}_{\infty}$, there exists $i$ such that
\begin{equation}\begin{split}\mathcal{P}_e(\Lambda,\HH) &\leq P(\mathbf{W} \notin \mathcal{B}_{\sqrt{mT(\sigma_w^2+\varepsilon)}}) + e^{\eta \alpha (mT(1+\varepsilon/\sigma_w^2))}\mathcal{P}_e(\Lambda, \HH_i | \mathbf{W} \in \mathcal{B}_{\sqrt{mT(\sigma_w^2+\varepsilon)}} ).
\end{split}
\end{equation} 
Taking the average over the ensemble:
\begin{equation} \begin{split}
E_{\mathbb{L}}\left[\mathcal{P}_e(\Lambda,\HH)\right] &\leq P(\mathbf{W} \notin \mathcal{B}_{\sqrt{mT(\sigma_w^2+\varepsilon)}}) +  (1+\delta)\lambda {D}^{-{T}} L_{\eta,n} e^{\eta \alpha (mT(1+\varepsilon/\sigma_w^2))} { \mbox{vol } \mathcal{B}_{\sqrt{mT(\sigma_w^2+\varepsilon)}}}.
\end{split}
\end{equation} 
If we choose $\lambda$ to be less than $D^{-T} / \sqrt{\pi e \sigma_w^2}$, $\lambda D^{-{T}}{ \mbox{vol } \mathcal{B}_{\sqrt{mT(\sigma_w^2+\varepsilon)}}}$ tends to $0 \mbox{ exponentially in } T.$
Therefore, we can choose $L_{\eta,m}$, independent of $T$, such that the total exponent is negative, and hence the average probability of error of the ensemble can be made arbitrarily small.
\end{proof}

We close this section arguing that mod-$p$ lattices \cite{Loeliger} fail to be universally good (for model \eqref{eq:compoundModel}). Suppose that $\HH$ is diagonal. All mod-$p$ lattices contain multiples of the canonical vectors (say, $ p \beta \mathbf{e}_i$, where $\beta$ is a scaling factor). Hence $\mathcal{V}_{\HH\Lambda}$ is contained in the set $S = \left\{ \mathbf{x} \in \mathbb{C}^{mT}: |x_1| \leq h_1 \beta p/2 \right\}$, and therefore for any $\Lambda$ in the mod-$p$ ensemble
\begin{equation}
P_e(\Lambda,\HH) \geq P(\mathbf{z} \notin S) = P(|z_1| \geq h_1 \beta p/2).
\end{equation}
Consider now the matrix $\HH \in \mathbb{H}_{\infty}$, with $h_1 = 1/p^2, h_2 = p^2$, $h_i = D^{1/(m-2)}, i = 3, \ldots, m$. It is clear that $P_e(\Lambda,\HH) \to 1$, as $p \to \infty$, and there is no good lattice (in the sense of Def. \ref{def:goodLattices}) in the ensemble. This does \textit{not} contradict \cite[Thm. 3]{HeshamElGamal04}, who showed, for a \textit{given fixed} $\mathbf{H}$, the existence of a good $\Lambda$ (depending on $\mathbf{H}$), which does not imply the existence of one single sequence with vanishing probabilities for all $\mathbf{H}$. We show later how to prevent this effect, by constructing lattices with full diversity.

\section{Shaping: The Lattice Gaussian Distribution}
\label{sec:finalScheme}
For the power-constrained model, the final transmission scheme is similar to \cite{LingBelfiore}. Using a coding lattice of dimension $mT$ from an ensemble satisfying Theorem \ref{thm:infiniteLimit}, the transmitter chooses a vector $\xx$ in $\Lambda$ drawn according to a lattice Gaussian distribution $D_{\Lambda,\sigma_s}$. The received applies MAP decoding to recover an estimate $\hat{\xx}$ of the sent symbol.

Consider a vector-form channel equation \eqref{eq:block-fading-vectorized}, with indices omitted for simplicity. Let $\rho = \frac{\sigma_s^{2}}{\sigma_w^{2}}$. MAP decoding reads: 
\begin{eqnarray*}
  \hat{\mathbf{x}} &=& \arg\max_{x\in\Lambda} p(\mathbf{x}|\mathbf{y}, \mathcal{H}) = \arg\max_{x\in\Lambda} p(\mathbf{y}|\mathbf{x}, \mathcal{H}) p(\mathbf{x})\\
   &{=}& \arg\max_{x\in\Lambda} f_{\sigma_w}(\mathbf{y}-\mathcal{H}\mathbf{x}) f_{\sigma_s}(\mathbf{x}) \\
   &\stackrel{(a)}{=}& \arg\min_{x\in\Lambda} \sigma_w^{-2}\|\mathbf{y}-\mathcal{H}\mathbf{x}\|^2 + \sigma_s^{-2}\|\mathbf{x}\|^2 \\
   &\stackrel{(b)}{=}& \arg\min_{x\in\Lambda} \|\mathbf{F}\mathbf{y}-\mathbf{R}\mathbf{x}\|^2
\end{eqnarray*}
where $(\mathbf{R},\mathbf{F})$ is any pair of matrices in $\mathbb{C}^{nT\times mT}$ satisfying $\mathbf{R}^\dagger\mathbf{R} = {\mathcal{H}^\text{tr}\mathcal{H}+\rho^{-1}\mathbf{I}}$ and $\mathbf{F}^\dagger \mathbf{R}={\rho}^{-1}\mathcal{H}$. In the above equation, (a) is due to the definition of $D_{\Lambda,\sigma_s}$, while (b) is obtained by completing the square. This coincides with the well-known MMSE-GDFE \cite{HeshamElGamal04}, except that $\mathrm{SNR}$ is replaced by $\rho$. We note that the matrices $\mathbf{F}$ and $\mathbf{R}$ are block diagonal, namely, the MMSE filter
is only applied on the spatial dimension. Therefore MAP decoding is equivalent to MMSE-GDFE filtering plus lattice decoding.

To analyze the error probability, we write
\begin{eqnarray*}
  \mathbf{y}' = \mathbf{Fy} = \mathbf{Rx} + (\mathbf{F}\mathcal{H}-\mathbf{R})\mathbf{x}+\mathbf{Fw} =\mathbf{Rx} + \mathbf{w}'
\end{eqnarray*}
where $\mathbf{w}' \triangleq  (\mathbf{F}\mathcal{H}-\mathbf{R})\mathbf{x}+\mathbf{Fw} $ can be viewed as the equivalent noise. The error probability of lattice decoding associated with $\Lambda$ is given by
\begin{eqnarray}
P_e(\Lambda) = \sum_{\mathbf{x} \in \Lambda}{P(\mathrm{error}|\mathbf{x})P(\mathbf{x})} 
   =P\left\{\mathbf{w}' \not\in \mathcal{V}(\mathbf{R}\Lambda)\right\} \label{eq:prob}
\end{eqnarray}
where the last step follows from the total probability theorem. We stress that in (\ref{eq:prob}), the probability is evaluated with respect to both distributions $\mathbf{x}\sim D_{\Lambda,\sigma_s}$ and $\mathbf{w}\sim f_{\sigma_w}$.

Next, we will show that the equivalent noise $\mathbf{w}'$ is sub-Gaussian. Therefore, the error probability is exponentially bounded above by that of a Gaussian noise, and a good infinite lattice coding scheme as in the proof of Thm. \ref{thm:infiniteLimit} will also have a vanishing probability of error for $\mathbf{w}'$. Let us recall the definition of sub-Gaussian random variables.

\begin{defi}[sub-Gaussian \cite{Vershynin11}]
A real-valued random variable $X$ is sub-Gaussian with parameter $\sigma > 0$ if for all $t \in \mathbb{R}$, the  moment-generating function satisfies
$\mathbb{E} [e^{tX}] \leq e^{\sigma^2 t^2/2}$. More generally, we say that a real random vector $\mathbf{X}$ is sub-Gaussian (of parameter $\sigma$) if all its one-dimensional marginals $\mathbf{u}^T\mathbf{X}$ for a unit vector $\mathbf{u}$ are sub-Gaussian (of parameter $\sigma$). We will say that a complex random vector is sub-Gaussian with parameter $\sigma$ if its real equivalent (under the transformation $\psi(\XX) = (\Re (\XX), \Im (\XX)$) is sub-Gaussian with parameter $\sigma/\sqrt{2}$.
\end{defi}

Note that the tails of a real-valued sub-Gaussian random variable $X$ are upper bounded the same way (satisfy the same Chernoff bound) that the tails of a normal distribution with parameter $\sigma$, i.e., $\mathbb{P}(|X| \geq t) \leq
2 e^{-t^2/(2\sigma^2)}$ for all $t \geq 0$.

\begin{lemma}
Let $\mathbf{x}\sim D_{\Lambda,\sigma}$. Then the moment generating function of $\mathbf{Ax}$ for any square matrix $\mathbf{A}$ satisfies
\[
E[e^{\Re(\mathbf{t}^{\dagger}\mathbf{Ax})}] \leq e^{\frac{\sigma^2}{4}\|\mathbf{A}^{\dagger}\mathbf{t}\|^2}.
\]
\end{lemma}

\begin{proof}
We rewrite the moment generating function as follows:
\begin{eqnarray*}
&&  f_{\sigma}(\Lambda) \cdot  E[e^{\Re(\mathbf{t}^\dagger\mathbf{Ax})}] = \frac{1}{(\sqrt{\pi}\sigma)^n}\sum_{\mathbf{x}\in\Lambda} e^{-\frac{\|\mathbf{x}\|^2}{\sigma^2}+\mathbf{t}^\dagger\mathbf{Ax}} \\
   &=& e^{\frac{\sigma^2}{4}\|\mathbf{A}^\dagger\mathbf{t}\|^2} f_{\sigma}\left(\Lambda-\frac{\sigma}{2}\mathbf{A}^\dagger\mathbf{t}\right),
\end{eqnarray*}
where the last inequality is obtained by ``completing the square''. Since $f_{\sigma}\left(\Lambda-\mathbf{a}\right) \leq f_{\sigma}(\Lambda) $ for any vector $\mathbf{a}$, the proof is completed.
\end{proof}

\begin{lemma}
The equivalent noise $\mathbf{w}'$ is sub-Gaussian with parameter $\sigma_w$.
\label{lem:sub-Gaussian}
\end{lemma}

\begin{proof}
	Let us derive its moment generation function:
	\begin{eqnarray*}
		E [e^{\Re(\mathbf{t}^\dagger\mathbf{w}')}] &=& E [e^{\Re(\mathbf{t}^\dagger((\mathbf{F}\mathcal{H}-\mathbf{R})\mathbf{x}+\mathbf{Fw}))}] \\
		&=& E [e^{\Re(\mathbf{t}^\dagger(\mathbf{F}\mathcal{H}-\mathbf{R})\mathbf{x})}]E[(e^{\Re(\mathbf{t}^\dagger\mathbf{Fw})}] \\
		&\leq& e^{\mathbf{t}^\dagger(\mathbf{F}\mathcal{H}-\mathbf{R})(\mathbf{F}\mathcal{H}-\mathbf{R})^\dagger\mathbf{t}\cdot \sigma^2_s/4 
			 } \cdot e^{ \mathbf{t}^\dagger\mathbf{F}\mathbf{F}^\dagger\mathbf{t}\cdot\sigma_w^2/4}\\
		&=& e^{\mathbf{t}^\dagger[\frac{\sigma^2_s}{4}(\mathbf{F}\mathcal{H}-\mathbf{R})(\mathbf{F}\mathcal{H}-\mathbf{R})^\dagger + \frac{\sigma_w^2}{4}\mathbf{F}\mathbf{F}^\dagger]\mathbf{t}}= e^{\frac{\sigma^2_c}{4}\|\mathbf{t}\|^2}.
	\end{eqnarray*}
	The last step holds because the covariance matrix \cite{HeshamElGamal04}
	\begin{eqnarray*}
		&&\sigma_s^2(\mathbf{F}\mathcal{H}-\mathbf{R})(\mathbf{F}\mathcal{H}-\mathbf{R})^\dagger + \sigma_w^2\mathbf{FF}^\dagger \\ & =& \sigma_s^2 \rho^{-2} \mathbf{R}^{-T}\mathbf{R}^{-1} + \sigma_w^2\rho^{-2}\mathbf{R}^{-T}\mathcal{H}^\dagger\mathcal{H}\mathbf{R}^{-1}\\
		&=& \sigma_w^2 \mathbf{R}^{-T} (\rho^{-1}\mathbf{I} + \mathcal{H}^\dagger\mathcal{H})\mathbf{R}^{-1}= \sigma_w^2\mathbf{I}.
	\end{eqnarray*}
	For any unit vector $\mathbf{u}$, we have
	\begin{eqnarray*}
		E [e^{\Re(t\mathbf{u}^\dagger\mathbf{w}')}]= E [e^{\Re((t\mathbf{u})^\dagger\mathbf{w}')}]
		&\leq& E [e^{\sigma^2_c\|t\mathbf{u}\|^2/4}]= e^{\sigma^2_c{t}^2/4}
	\end{eqnarray*}
	completing the proof.
\end{proof}
Finally, from Theorem \ref{thm:infiniteLimit}, taking a universal lattice $\Lambda$ from the Minkowski-Hlawka ensemble \eqref{eq:ensemble}, the error probability vanishes as long as the VNR $\gamma_{\mathbf{R}\Lambda}(\sigma) >\pi e$ (as $T\to \infty$), i.e..
\begin{equation}\label{eq:VNR}
\frac{V(\mathbf{R}\Lambda)^{\frac{1}{mT}}}{ \sigma_w^2} = \frac{ |\rho^{-1}\mathbf{I} + \mathbf{H}^\dagger\mathbf{H}|^{\frac{1}{m}}V(\Lambda)^{\frac{1}{mT}}}{ \sigma_w^2} > \pi e.
\end{equation}
Thus,  from \cite[Lemma 6]{LLBS_12}, any rate
\begin{equation}
\begin{split}
R &= m\log(\pi e \sigma_s^2) - \frac{1}{T} \log(V(\Lambda)) -\varepsilon \\&
= \log \det \left(\mathbf{I} + \rho\mathbf{H}^\dagger\mathbf{H}\right) - \varepsilon = C-\varepsilon,
\end{split}
\end{equation}
for any arbitrarily small $\varepsilon$ is achievable. Note that the achievable rate only depends on $\mathbf{H}$ through $\det(\mathbf{I} + \rho\mathbf{H}^\dagger\mathbf{H})$. Therefore, there exists a lattice $\Lambda$ achieving capacity $C$ of the compound channel.

The techniques above greatly simplify the probability of error analysis in \cite{LingBelfiore}. Note that, for the probability of error, we do not need a flatness condition on the distribution as in \cite{LingBelfiore} anymore, thanks to sub-Gaussianity.\footnote{However, contrary to what was stated in a previous version of this paper \cite{OurISIT}, we do need flatness of $D_{\Lambda,\sigma_s^2}$ for the entropy approximation. More precisely, we need $\epsilon_{\Lambda}(\sigma_s)$ to be negligible, which can be satisfied above a threshold $\text{snr} > e-1$. When this is the case, the signal power $P\approx \sigma^2$ and $\rho \approx \mathrm{SNR}$. The threshold can be further reduced to zero by shaping over a random coset of $\Lambda$, or by constructing $\Lambda$ using the methods in Section \ref{sec:ConsAAlgebraic} with a random coset of a capacity-achieving linear code. The details are out of the scope of this paper and thus omitted.}

\section{Construction of Good Ensembles}
\label{sec:consA}
\noindent Theorem \ref{thm:infiniteLimit} ultimately relies on the existence of an ensemble of lattices satisfying two conditions:
\begin{enumerate}
\item The Minkowski-Hlawka Theorem (Theorem \ref{thm:basicMH}) and
\item The compactification property (Definition \ref{def:compacify}).
\end{enumerate} 
In this section, we show how to construct lattices with these properties. Our main tool is Algebraic Number Theory, previously used to develop good modulation schemes for MIMO and fading channels \cite{ViterboOggier}. We separate two cases: The block-fading case (where the channel matrix $\mathbf{H}$ is diagonal), and the MIMO case (for general $\mathbf{H}$). Although the latter case contains the former, block-fading channels are of independent interest and have special commutative structures that can be exploited to simplify the code construction and analysis.

\subsection{Ensembles for the Block-Fading Channel}
\label{sec:ConsAAlgebraic}
We follow closely the construction of \cite{KosiOngOggier}, also used in \cite{Adaptative} for the Compute-And-Forward protocol. For an introduction to the algebraic theory used in this section, the reader is referred to \cite{ViterboOggier}. We describe in the next subsection some main concepts and results used throughout the paper.
\subsubsection{Basic Notation}
We consider \textit{(algebraic) number fields} $K/\mathbb{Q}$, i.e. field extensions of $\mathbb{Q}$ with finite degree $\tilde{m}$. There are $\tilde{m}$ homomorphisms $\sigma_1,\ldots, \sigma_{\tilde{m}}$ that embed $K$ into $\mathbb{C}$ and fix $\mathbb{Q}$. If \textit{none} of the images of these embeddings is contained in $\mathbb{R}$, we say that $K$ is a \textit{totally complex} extension (as opposed to \textit{totally real}, when all images are in $\mathbb{R}$). From now on, unless stated otherwise, we assume that number fields are totally complex. In this case, $\tilde{m} = 2m$ is even and the homomorphisms $\sigma_i$ appear in complex conjugate pairs, i.e. we can assume that the homomorphisms are
$$\sigma_1,\overline{\sigma}_1, {\sigma}_2,\overline{\sigma}_2,\ldots, {\sigma}_{{m}}, \overline{\sigma}_{{m}}.$$ 

The \textit{ring of integers} of $K$, denoted by $\mathcal{O}_K$, is the ring of all elements in $K$ which are root of a monic polynomial with integer coefficients. The invertible elements in $\mathcal{O}_K$ are called \textit{units}. The mapping
$$\sigma:K\to\mathbb{C}^m$$
$$\sigma(x) = (\sigma_1(x),\ldots,\sigma_m(x))$$
is called the \textit{canonical embedding}. It takes $\mathcal{O}_K$ into a lattice in $\mathbb{C}^m$. Let $V$ be the volume of this lattice. The \textit{discriminant} of number field $K$ is given by $\Delta_K = (2^m V)^2$. 

Any ideal $\mathfrak{p} \subset \mathcal{O}_K$ can be decomposed as the product of prime ideals.  Let $p$ be a prime number and consider the decomposition 
$$p\mathcal{O}_K= \prod_{i=1}^g\mathfrak{p}_i^{e_i}.$$

We say that each $\mathfrak{p}_i$ is \textit{above} $p$. It follows that $\mathcal{O}_K/\mathfrak{p}_i \simeq \mathbb{F}_{p^l}$, for some $l$. When $g=2m$, $l=1$, and we say that $p$ \textit{splits}.

\begin{exe} Complex quadratic fields have the form $\mathbb{Q}(\sqrt{d}) = \left\{ a + b \sqrt{d}, a, b \in \mathbb{Q} \right\}$, where $d < 0$ is a square-free number. Their ring of integers is $\mathbb{Z}[\sqrt{d}]$ if $d \equiv 1 \pmod 4$ or $\mathbb{Z}[(-1+\sqrt{d})/2]$ if $d \equiv -1 \pmod 4$.
Adjoining $\sqrt{c} > 0$, where $c > 0$ is a real square-free number, the field 
$$\mathbb{Q}(\sqrt{c},\sqrt{d}) = \left\{ a_1 + a_2 \sqrt{c} + a_3 \sqrt{d} + a_4 \sqrt{dc} \right\}$$ is a totally complex extension of degree $4$ (quartic). These are called bi-quadratic number fields.
\label{ex:CM}
\end{exe}
\begin{exe}
	The special case $K=\mathbb{Q}(i,\sqrt{5})$ was previously used to construct the so-called Golden Code \cite{Golden} for transmission over a $2\times2$ MIMO channel. Let $\theta=(1+\sqrt{5})/2$. The ring of integers $\mathcal{O}_K\subset K$ is generated, as $\mathbb{Z}[i]$-module, by $\left\{1,\theta \right\}$. The prime $3$ splits into the product of two prime ideals, as can be seen by:
	$$3= ((i-1)\theta+1)(-(i+1)\theta+1).$$
	If $\mathfrak{p}=((i-1)\theta+1)\mathcal{O}_K$, then the quotient $\mathcal{O}_K/\mathfrak{p} \simeq \mathbb{F}_{3^2}$. The first prime that splits completely is $29$, namely
	$$29=(-i\theta+2)(-i\theta+i-2)(i\theta+2)(i\theta-i-2).$$
	The quotient of $\mathcal{O}_K$ by the ideal generated by any of its factors is isomorphic to $\mathbb{F}_{29}$.
	
	\label{ex:golden}
\end{exe}

\subsubsection{Construction A}\label{subsec:consACompound} Let $\mathfrak{p} \subset \mathcal{O}_K$ be a prime ideal above $p$, so that there exists an isomorphism $\phi:\mathcal{O}_K/\mathfrak{p}\to \mathbb{F}_{p^l}$. Denote by $\pi$ the canonical projection $\pi: \mathcal{O}_k \to \mathcal{O}_K/\mathfrak{p}$. We also use the ``overloaded'' notation $\pi$ and $\phi$ to denote the componentwise transformations applied to the cartesian products $\mathcal{O}_K^T$ and $(\mathcal{O}_K/\mathfrak{p})^T$.

Now let $\mathcal{C} \subset\mathbb{F}_{p^l}^T$ be a linear $(T,k)$-code, i.e, a subspace of $\mathbb{F}_{p^l}^T$ with dimension $k$. The $\mathcal{O}_K$-lattice associated to $\mathcal{C}$ is defined as the pre-image by $\phi \circ \pi$ of $\mathcal{C}$ ($\phi$ and $\pi$ are applied componentwise):
\begin{equation}\Lambda^{\mathcal{O}_K}(\mathcal{C}) = \pi^{-1}\circ \phi^{-1}(\mathcal{C}).
\end{equation}
If $\mathcal{C}$ is linear, $\Lambda^{\mathcal{O}_K}(\mathcal{C})$ is a lattice and $\Lambda^{\mathcal{O}_K}(\mathcal{C})/\mathfrak{p}^T\simeq \mathcal{C}$. The associated complex lattice $\Lambda(\mathcal{C})$ is obtained by applying (elementwise) the canonical embedding ${\sigma:K\to \mathbb{R}^{n}}$. It follows that an element $y=\sigma(\mathbf{x})$, with $\mathbf{x} \in \mathcal{O}_{K}^T$, belongs to  $\Lambda(\mathcal{C})$ if and only if $(\phi \circ \pi)( \mathbf{x} )\in \mathcal{C}$. The matrix form representation of a lattice point is:
\begin{equation}\mathbf{X}= \left( \begin{array}{cccc} \sigma_1(x_1) & \sigma_1(x_2) & \cdots & \sigma_1(x_T) \\ \sigma_2(x_{1}) & \sigma_2(x_{2}) & \cdots & \sigma_2(x_{T}) \\ \vdots & \vdots & \ddots & \vdots \\  \sigma_n(x_{1}) &\sigma_m(x_{2}) & \cdots & \sigma_m(x_{T})\end{array} \right).
\end{equation}

\begin{prop} A lattice $\Lambda(\mathcal{C})$ be a constructed as above has the following properties
\begin{enumerate}
\item $V(\Lambda(\mathcal{C})) = 2^{-mT} p^{{l(T-k)}} (\sqrt{\Delta_K})^T$
\item If $u$ is a unit in $\mathcal{O}_K$ and $\mathbf{U} = \text{\upshape{diag}}(\sigma_1(u),\ldots,\sigma_m(u))$, then, in matrix form, $\mathbf{U}\Lambda = \Lambda$.
\end{enumerate}
\label{prop:VolumeAndUnit}
\end{prop}
%
%

Property 2) says that Generalized Construction A lattices are closed by multiplication by units, i.e., if $u$ is a unit in $\mathcal{O}_k$, then $u \Lambda^{\mathcal{O}_K}(\mathcal{C}) = \Lambda^{\mathcal{O}_K}(\mathcal{C})$. This is a crucial property for proving the compactification property.

It is proven in Appendix \ref{app:numberFields}, following steps of \cite{Loeliger} and \cite[Appendix B]{Adaptative}, that the set of such lattices satisfies, asymptotically, the Minkowski-Hlawka theorem, as $p \to \infty$. More formally, let $\lambda > 0$ be a scaling factor and $\alpha= (\lambda^{-1} 2^{-mT} p^{{l(T-k)}} \sqrt{\Delta_K}^T)^{1/2mT}$. Consider a bounded function $f:\mathbb{R}^{2mT} \to \mathbb{R}$ with compact support. The ensemble
\begin{equation}\begin{split}
\mathbb{L}_{K,T,k,p,\lambda} = \left\{ \frac{1}{\alpha} \Lambda_K(\mathcal{C}) : \mathcal{C} \mbox{ is a }(T,k,p)\mbox{ code }\right\}
\end{split}
\label{eq:ensemble}
\end{equation}
satisfies
\begin{equation}\lim_{p \to \infty} E_{\mathbb{L}_{K,T,k,p,\lambda}} \left[ \sum_{\xx \in \Lambda(\mathcal{C})\backslash\left\{\mathbf{0}\right\} }  f(\psi(\xx)) \right] = \lambda \int_{\mathbb{R}^{2mT}} f(\mathbf{x}) \mbox{d} \mathbf{x},
\label{eq:ensembleAverage}
\end{equation}
where we recall that $\psi(\xx) = (\Re(\xx), \Im(\xx))$ is the standard identification between $\mathbb{C}^{mT}$ and $\mathbb{R}^{2mT}$. All lattices in the ensemble have volume $1/\lambda$.

\subsubsection{Quantizing the Channel Coefficients}
We use the group of units of $\mathcal{O}_K$ to quantize the channel coefficients for the compound block-fading channel. The main tool is the group of units and the invariance property given in Proposition \ref{prop:VolumeAndUnit}.(ii). In the block-fading channel, $m = n$ and the matrix $\HH$ is diagonal. We define
\begin{equation}
\mathbb{H}_{\infty}^\text{BF} \triangleq \left\{\HH \in \mathbb{R}^{n\times n}: \HH \mbox{ is diagonal and } \det \HH^{\dagger}\HH= D \right\}
\label{eq:BFcompoundModel}
\end{equation}
and $\tilde{\mathbb{H}}_{\infty}^\text{BF} =\mathbb{H}_{\infty}^\text{BF}/D^{1/2n}$. Let $\mathbf{U} = \mbox{diag}(\sigma_1(u),\ldots,\sigma_n(u) )$ be the diagonal matrix corresponding to the embedding of a unit. Let $\mathcal{U}$ be the set of all possible matrices $\mathbf{U}$. For a normalized channel matrix $\tilde{\mathbf{H}}\in \tilde{\mathbb{H}}_{\infty}^\text{BF}$, we define
\begin{equation}\mathbf{U}_{\tilde{\HH}} = {\underset{\mathbf{U} \in \mathcal{U}}{\arg \min}}\left\| \tilde{\mathbf{H}} \mathbf{U}^{-1}\right\|,
\end{equation}
with ties broken in a systematic manner. The association $\mathbf{\tilde{\HH}} \to \mathbf{U}_{\tilde{\HH}}$ defines an equivalence relation. By quotienting $\tilde{\mathbb{H}}_\infty$ by this relation, we obtain the equivalence classes associated to the error matrices $\mathbf{E}_H = \tilde{\HH }\UU_{\tilde{\HH}}^{-1}$. Let

\begin{equation}\mathcal{E} = \left\{\mathbf{E}_{\tilde{\HH}} = \tilde{\HH }\UU_{\tilde{\HH}}^{-1}: \tilde{\HH} \in \tilde{\mathbb{H}} \right\}
\label{eq:errorMatrix}
\end{equation}
be the set of all the possible error matrices. In what follows we argue that $\mathcal{E}$ is compact and thus the ensemble $\mathbb{L}_{K, T,k,p,\lambda} $ (Eq. \eqref{eq:ensemble}) compactifies $\tilde{\mathbb{H}}_{\infty}^{\text{BF}}$ as in Definition \ref{def:compacify}. We also provide bounds on length of the elements of $\mathcal{E}$. First recall that Dirichlet's Unit Theorem (e.g. \cite[Thm 7.3]{Algebraic1}) states the existence of $u_1,\ldots,u_{n-1}$ fundamental units such that any unit in $\mathcal{O}_K$ can be written as 
\begin{equation}u = \zeta \prod_{i=1}^{n-1} u_i^{k_i}\mbox{, where } k_i \in \mathbb{Z} \mbox{ and } \zeta \mbox{ is a root of unit.}
\end{equation}
%
This implies that the group of units, under the transformation
\begin{equation}\begin{split} \ell(u) = (\log |\sigma_1(u)|^2,\ldots,\log |\sigma_n(u)|^2)
\end{split}
\end{equation}
is an $(n-1)$-dimensional lattice in $\mathbb{R}^{n}$, contained in the hyperplane orthogonal to the vector $(1,\ldots, 1)$. The volume of this lattice, referred to as \textit{logarithmic lattice}, is called the \textit{regulator} of $K$. 

 \begin{thm} For any channel matrix $\HH$, there exists $\UU =  \mbox{diag}(\sigma_1(u),\ldots,\sigma_n(u) )$ such that 
 $$\left\| \tilde{\HH} \UU^{-1} \right\| \leq \sqrt{n} e^{\rho}, $$
 where $\rho$ is the packing radius of the logarithmic lattice.
 \label{thm:boundLogLattice}
 \end{thm}
 \begin{proof} Write the magnitudes of the diagonal elements of $\tilde{\HH}$ in vector form as $\tilde{\mathbf{h}} = (\log|\tilde{h}_1|^2,\ldots, \log|\tilde{h}_n|^2)$. Let $\mathbf{v} = (\log|\sigma_1(u)|^2,\ldots, \log|\sigma_n(u)|^2)$ be the closest point in the logarithm lattice to $\tilde{\mathbf{h}}$. Let $\rho$ be the covering radius of the logarithmic lattice. We have $(\log|\tilde{h}_i| - \log|\sigma_i(u)|) \leq  \left\|\tilde{\mathbf{h}}-\mathbf{v} \right\| \leq \rho$, therefore
 \begin{equation}\begin{split}
\sum_{i=1}^n |\tilde{h}_{i}|^2 |\sigma_i(u)^{-1}|^2 = \sum_{i=1}^n e^{( \log| \tilde{h}_i|^2- \log|\sigma_i(u)|^2)} \leq n e^{2\rho}.
\end{split}
\label{eq:rhoLogLattice}
\end{equation}
 \end{proof}
 \begin{rmk} If $n = 2$ (quartic extension), the logarithmic lattice is one-dimensional, therefore $\rho = R_K/2$, where $R_K$ is its regulator. In this case, a tight estimate for the norm of the error is
$$\left\| \tilde{\HH} \UU^{-1} \right\| \leq \sqrt{2 \cosh(\rho)},$$
which is achieved when $\tilde{h}_1 = R_K/2$.
 \end{rmk}
\begin{cor} The ensemble $\mathbb{L}_{K, T,k,p,\lambda}$ (Eq. \eqref{eq:ensemble}) compactifies $\tilde{\mathbb{H}}_{\infty}^{\text{BF}}$, in the sense of Definition \ref{def:compacify}.
\end{cor}
\begin{exe}For the sake of exemplification, consider a real fading channel and the totally real number field $K = \mathbb{Q}[\sqrt{5}]$, so that $\mathcal{O}_K=\mathbb{Z}[\phi],$ where $\phi = (1+\sqrt{5})/2$ is the Golden ratio. The units of $\mathcal{O}_K$ are of the form $\pm \phi^k$, $k \in \mathbb{Z}$, and its embeddings in $\mathbb{R}^2$ are the blue dots depicted in Figure \ref{fig:1}. After normalization, the channel realizations $h_1,h_2$ lie in the hyperbola $h_1 h_2 = 1$. Any realization $(h_1,h_2)$ can be taken, by multiplication by an appropriate unit, to a bounded fundamental domain. This way, ill-conditioned channel realizations can be ``absorbed'' by the group of units.

\begin{figure}
\centering
\includegraphics[scale=0.4]{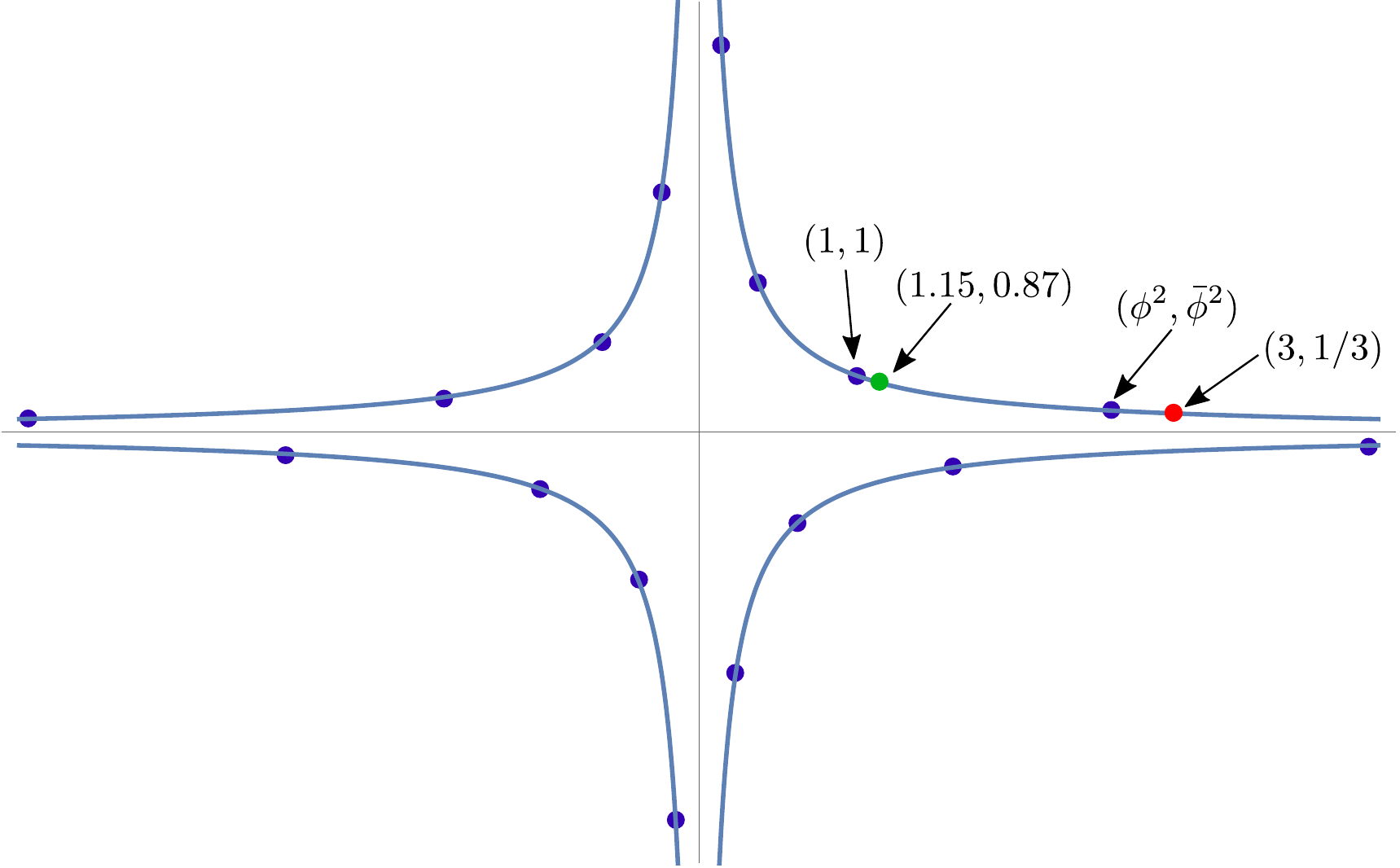}
\label{fig:1}
\caption{Handling an ill-conditioned channel realization by the quantization of the channel space}
\end{figure}
\label{ex:Q5}
\end{exe} 

\subsection{Division Algebras}
The extension to the MIMO case entails a construction based on division algebras. In this subsection, we follow closely the notation and construction of \cite[Sec. V]{ConsAOggier}. 

\subsubsection{Relative Extensions}\label{subsec:relativeExtensions}Consider the field extension $K/\mathbb{Q}(i)$ of relative degree $m$ (i.e., the absolute extension $K/\mathbb{Q}(i)/\mathbb{Q}$ has absolute degree $2m$). Suppose that the Galois group of $K/\mathbb{Q}(i)$ is cyclic. This means that the $m$ embeddings that fix $\mathbb{Q}(i)$ can be generated by one element, say, $\beta$. Let 
$$\sigma(x) = (\alpha^0(x) = x,\beta(x),\beta^2(x),\ldots,\beta^{m-1})$$
 be the cannonical embedding of an element in $\mathbb{Q}(i)$. 


\subsubsection{Algebras}
A cyclic division algebra $\mathcal{A}=K\oplus eK\oplus\cdots\oplus e^{m-1}K$, denoted by $\mathcal{D} = (K/\mathbb{Q}(i), \beta, \gamma)$, is the algebra of all elements  $a=x_0 + x_1 e + \ldots + x_{m-1} e^{m-1}$, where $e$ is an element such that $e^{n} = \gamma$, and $\gamma \in \mathbb{Z}[i]$ is an element which is \textit{not} a (relative) norm in $K/\mathbb{Q}(i)$. Multiplication is done by the rule $xe = e\beta(x)$. An element $a$ in the division algebra can be represented in matrix form as
\begin{equation}
\XX_a= \left(\begin{array}{ccccc}
x_0 & \gamma \beta(x_{m-1}) & \gamma \beta^2(x_{m-2}) & \ldots & \gamma\beta^{m-1}(x_1)\\
x_1 & \beta(x_0) & \gamma \beta^2(x_{m-1}) & \ldots & \gamma\beta^{m-1}(x_2) \\
\vdots & \vdots & \ddots & \vdots & \vdots\\
x_{m-1} & \beta(x_{m-2}) & \ldots & \ldots & \beta^{m-1}(x_0)
\end{array}\right)
\end{equation}
and multiplication corresponds to the standard matrix multiplication. From now on, we consider this representation.

Lattices from cyclic algebras are constructing using orders. The set $\Lambda = \mathcal{O}_K\oplus e\mathcal{O}_K \oplus\cdots\oplus e^{m-1}\mathcal{O}_K$ is called the \textit{natural order} of $\mathcal{A}$. Let $p$ be a prime that splits in $\mathbb{Z}[i]$, and keep the notation as in Section \ref{subsec:relativeExtensions}. Let 
$\mathbb{F}_p^{m\times m}$ be the set of all $m \times m$ matrices with entries in $\mathbb{F}_p$. With a small abuse of notation, define the reduction 
$$(\phi \circ \pi) : \Lambda \to \mathbb{F}_p^{m\times m},$$ 
$$(\phi \circ \pi)(a) = \phi\circ\pi(\mathbf{X}_a),$$
to be the componentwise reduction in $\mathbf{X}_a$.

Now extend all mappings to $T$-uples of elements, i.e,
$$(\phi \circ \pi)(a_1,\ldots,a_T) = ((\phi \circ \pi)(\mathbf{X}_{a_1}),(\phi \circ \pi)(\mathbf{X}_{a_2}),\ldots,(\phi \circ \pi)(\mathbf{X}_{a_T})).$$

\begin{defi} A \textit{linear} code $\mathcal{C}$ over the matrix ring $\mathbb{F}_p^{m\times m}$ with length $T$ is a subset of $(\mathbb{F}_p^{m \times m})^T$ which is closed under addition.
\end{defi}

 Let $\mathcal{C}$ be a code in $(\mathbb{F}_p^{m\times m})^T$. Then $\Lambda_p(\mathcal{C}) = \beta^{-1}(\mathcal{C})$. If the code is linear, $\Lambda_p(\mathcal{C})$ is a lattice (with complex equivalent in $\mathbb{C}^{t^2m}$). The volume of the equivalent (vectorized) lattice in $\mathbb{C}^{t^2m}$ is given by $\det \Lambda_p(\mathcal{C}) = |C|^{-1} p^{t^2m} (2^{-m}\gamma^{m(m-1)/2} \sqrt{\Delta_K^t})^m$.



In Appendix \ref{subsec:DivisionAppendix}, we prove that there exists a good ensemble of lattices from the aforementioned construction. In addition, in \cite[Thm. 1, p. 214]{Kleinert} it is proven that it compacifies the space $\mathbb{H}_{\infty}$. Therefore, the ensemble achieves the capacity of the infinite model.

\section{Decoupling Technique}
\label{sec:decoupled}
In order to recover the sent lattice point in block-fading and MIMO channels, the receiver usually performs a universal lattice decoder (such as the sphere decoder). This is due to the fact that, even if the coding lattice $\Lambda$ is well-structured and has a good decoding algorithm, the channel realization $\HH$ is arbitrary, forcing the receiver to decode in the modified lattice $\HH \Lambda$, which increases significantly the decoding complexity. 

To overcome this problem, Ordentlich and Erez \cite{Ordentlich15} consider the notion of \textit{decoupling}. A decoupled decoder first handles the channel realization $\HH$ and then decodes using the lattice-decoding algorithm in $\Lambda$ itself. This process is sub-optimal and produces a gap to capacity. As long as the gap is constant (as is the case in \cite{Ordentlich15}), decoupling can be an interesting alternative in the high SNR regime. In what follows, we show that the algebraic ensembles defined in the previous section allow for decoupling and calculate the gap to capacity. This technique appeared previously in the literature, in practical modulation schemes for the Rayleigh fading channel \cite{GhayaViterboJC} and for $2 \times 2$ MIMO channels in \cite{OthmanLuzziBelfiore10}.

Consider the received signal $\yy = \mathcal{H} \xx+ \zz$ and the MMSE filtering as in Section \ref{sec:finalScheme}. Notice again that $\RR$ is block-diagonal and $\det(\RR^\dagger \RR) = \rho^{-mT}e^{TC}$, where $C$ is the capacity of the compound channel. Let $D=\rho^{1/2}e^{-C/2m}$, and $\tilde{\RR} = D \RR$, so that every component in the matrix-diagonal of $\tilde{\RR}^\dagger \tilde{\RR}$ is a matrix with unit absolute determinant. Decoding consists in essentially three steps.

\begin{enumerate}
	\item Filtering: Apply a filtering matrix $\FF$ to obtain $\FF\yy = \RR\xx + \ww_{\text{eff}}$, where $\ww_{\text{eff}}$ is the effective noise.
	\item Equalization: Find $\EE_{\tilde{\HH}}$ as in Definition \ref{def:compacify} such that $\EE_{\tilde{\HH}} \UU = \tilde{\RR}$.
	\item Lattice decoding: Let $\tilde{\yy} = D^{-1} \EE_{\tilde{\HH}}^{-1} \FF\yy$. Find
	$$\arg \min_{\tilde{\xx} \in \Lambda} \left\|  \tilde{\yy}  -\tilde{\xx} \right\|$$
	and set $\xx = \UU^{-1} \tilde{\xx}$.
\end{enumerate}
%
Let $\tilde{\ww}_{\text{eff}} = D^{-1} \EE_{\tilde{\HH}}^{-1}\ww_{\text{eff}}$. The probability of error of such decoder is given by $P(\tilde{\ww}_{\text{eff}} \notin \mathcal{V}_{\Lambda}).$ Now if the decoder takes into consideration the correlations in $\EE_{\tilde{\HH}}^{-1}$, nothing is gained with the equalization task. The key observation is that the receiver can \textit{ignore} the correlations by performing lattice decoding in a slightly worse channel. This is formalized in the next proposition.
\
\begin{lemma}
Suppose that $\left\| \EE_{\tilde{\HH}}^{-1} \right\| \leq \alpha$, similarly to Definition \ref{def:compacify}. Then $\EE_{\tilde{\HH}}^{-1}\ww_{\text{\upshape{eff}}}$ is sub-Gaussian with parameter $\alpha\sigma_w$.
\end{lemma} 
\begin{proof}
	From Lemma \ref{lem:sub-Gaussian}, for a unitary vector $\uu$ we have
	\begin{equation*}
	\begin{split}
	E\left[e^{\Re \left( t\uu^\dagger  \EE_{\tilde{\HH}}^{-1} \ww_{\text{\upshape{eff}}}\right)}\right] &= E\left[e^{\Re \left( t\left\|\EE_{\tilde{\HH}}^{-\dagger} \uu\right\| \left(\frac{\uu^\dagger}{\left\|\EE_{\tilde{\HH}}^{-\dagger} \uu\right\|}  \EE_{\tilde{\HH}}^{-1} \ww_{\text{\upshape{eff}}}\right)\right)}\right]\\
	& \leq e^{\sigma_w^2 t^2 \left\|\EE_{\tilde{\HH}}^{-\dagger} \uu\right\|^2/4} \leq e^{\sigma_w^2 t^2 \left\|\EE_{\tilde{\HH}}^{-\dagger}\right\|^2/4} \leq e^{\alpha^2 \sigma_w^2 t^2/4},
	\end{split}
	\end{equation*}
	where $\left\| . \right\|$ denotes the Frobenius matrix norm, which upper bounds the operator $2$-norm.
	\end{proof}

Now from a Minkowski-Hlawka ensemble, the probability of error can vanish as long as

\begin{equation*}
\frac{D V(\Lambda)^{1/mT}}{\sigma_w^2 \alpha^2} \geq (\pi e).
\end{equation*}
and rates up to $R = m\log(\pi e \sigma_s^2) - \frac{1}{T} \log(V(\Lambda)) = C - 2 \log \alpha$ are achievable.
\begin{thm} The gap to the capacity of the decoupled decoder with notation as above is upper bounded by $2 \log \alpha$ nats per channel use.
	\label{thm:gap-to-Capacity}
\end{thm}
In light of Theorem \ref{thm:boundLogLattice}, we have the following corollary that characterizes the gap in terms of algebraic properties of the ensemble in a block fading channel. Notice that the equalization step, in this case, can be accomplished by lattice decoding in a logarithmic lattice of dimension $m$.
\begin{cor} For the block fading channel, the gap is upper bounded by $\log(m) + 2 \rho$, where $\rho$ is the covering radius of the logarithmic lattice.
\end{cor}
Since $m$ and $\rho$ do not depend on $T$, this gives a constant gap to capacity. For $m = 2$ (quartic fields), logarithmic lattices with small regulators are classified in \cite{RegulatorFormulas}. The minimum gap is $0.707308$ nats per channel use.
\begin{cor}
	For a $2\times2$ compound MIMO channel, the gap to capacity is upper bounded by $\approx 1.49784$ nats per channel use.
\end{cor}
\begin{proof}
	Combine Theorem \ref{thm:gap-to-Capacity} with \cite[Prop. 1]{OthmanLuzziBelfiore10}.
\end{proof}

In \cite{OthmanLuzziBelfiore10} the authors show an efficient method to accomplish the equalization step. Again, notice that this involves an algorithm whose complexity depends only on $m$, which is typically smaller than $T$ (for the capacity-achieving schemes $m$ is fixed whereas $T \to \infty$).
\section{Ergodic Channels}
\label{sec:Ergodic}
We show next how to extend the previous results for ergodic channels, where the channel coefficient vary according to a random process. The channel is described by equation
\begin{equation}
y_i = h_i x_i + z_i
\end{equation}
for $i=1,2,\ldots$,  where $z_i$ is a Gaussian noise $\sim \mathcal{CN}(0,\sigma_w^2)$ and $\left\{ h_i \right\}$ is a stationary ergodic random process with $E[|h_i|^2]=1$. The input values have average power lesser or equal than $P$. Let $\text{SNR} \triangleq P/\sigma_w^2$. The capacity of this channel is \cite{WirelessCommunication}
\begin{equation*}
C=E\left[\log\left(1+|h|^2\text{SNR}\right)\right] \mbox{ nats/channel use}.
\end{equation*}
and it is known to be achievable with random codes. A special case is when the fading coefficients are independent and identically distributed. 


\subsection{The Random Ensemble}
\subsubsection{Construction A}
Here we present an algebraic Construction A suitable for the ergodic fading model. This construction differs slightly from the one in \ref{subsec:consACompound} and was firstly studied in \cite{ConsAOggier}.

Consider a relative extension $K/\mathbb{Q}(i)$, a prime $p$, $\mathfrak{p}$ and $\mathfrak{b}$ as in Section \ref{subsec:relativeExtensions}.
Consider the projection $\pi: \mathcal{O}_K \to \mathcal{O}_K/\mathfrak{\beta}$. Each relative embedding $\sigma$ takes $\mathcal{O}_K$ into $\mathcal{O}_K$. Let $\Lambda_{K} = \sigma(\mathcal{O}_K)$ and consider reduction mapping $\Lambda_{K }\to  \mathbb{F}_p^n$ given by
\begin{equation}
\rho(\sigma(x)) = (\phi \circ \pi)(\sigma(x)),
\label{eq:mappingrho}
\end{equation}
where $(\phi \circ \pi)$ is applied component-wise in the canonical embedding. Then  Let $\mathcal{C} \subset \mathbb{F}_p^n$ be a linear code with dimension $k$ (or a code with parameters $(n,k,p)$). The Construction A lattice associated to $\mathcal{C}$ is defined as
$$\Lambda_{{K}}(\mathcal{C}) \triangleq \rho^{-1}(\mathcal{C}).$$
The properties of $\Lambda_{K}{(\mathcal{C})}$ are studied in \cite{ConsAOggier}. First of all $\Lambda_K(\mathcal{C})$ is a full rank lattice and $\Lambda_{{K}}(\left\{ 0 \right\}) = \sigma(p \mathcal{O}_K)$ is a sublattice with index $|\mathcal{C}| = p^k$. In fact the quotient
\begin{equation}
{\Lambda_K(\mathcal{C})}/{p \Lambda_K} \simeq \mathcal{C},
\end{equation}
from where we can deduce that $V(\Lambda_K(\mathcal{C}))=2^{-n} p^{n-k} \sqrt{\Delta_K}$. It is further shown that $\Lambda_{{K}}(\mathcal{C})$ has full-diversity and its product distance is bounded in terms of the Hamming distances of $\mathcal{C}$.

The above construction is very similar to the one in \ref{subsec:consACompound}  but there is a fundamental distance. In this case the dimension of the lattice is equal to the length of the underlying code, while in Construction \ref{subsec:consACompound} the dimension of the lattice is $nT$. In other words, while in the compound case the degree of the relative extension is fixed for the whole transmission, in the ergodic fading it increases with the block-length.

\begin{exe} For the sake of illustration consider a totally real number field, and the corresponding lattice obtained by the process above. Let $\mathbb{Q}[\sqrt{13}]$ be the quadratic field with ring of integers $\mathbb{Z}\left[\mu\right],$ where $\mu= \frac{1+\sqrt{13}}{2}$ and $\bar{\mu} = \frac{1-\sqrt{13}}{2}$. The two embeddings are determined by $\sigma_1(\sqrt{13})=\sqrt{13}$ and $\sigma_2(\sqrt{13})=-\sqrt{13}$. The prime $3 = - \mu \bar{\mu}$ splits and the ideal $\mathfrak{p}=\mu\mathbb{Z}[\mu]$ is such that $\mathbb{Z}[\mu]/\mathfrak{p} \sim \mathbb{F}_{3}$. We can now identify the set of representatives for the quotient $\mathbb{Z}[\mu]/3\mathbb{Z}[\mu]$ 
	with elements in $\mathbb{F}_3^2$, as in Figure \ref{fig:1}. The pre-image by $\rho$ of a code spreads its corresponding representatives in the plane.
	
	\begin{figure}[!htb]
		\centering
		\includegraphics[scale=0.35]{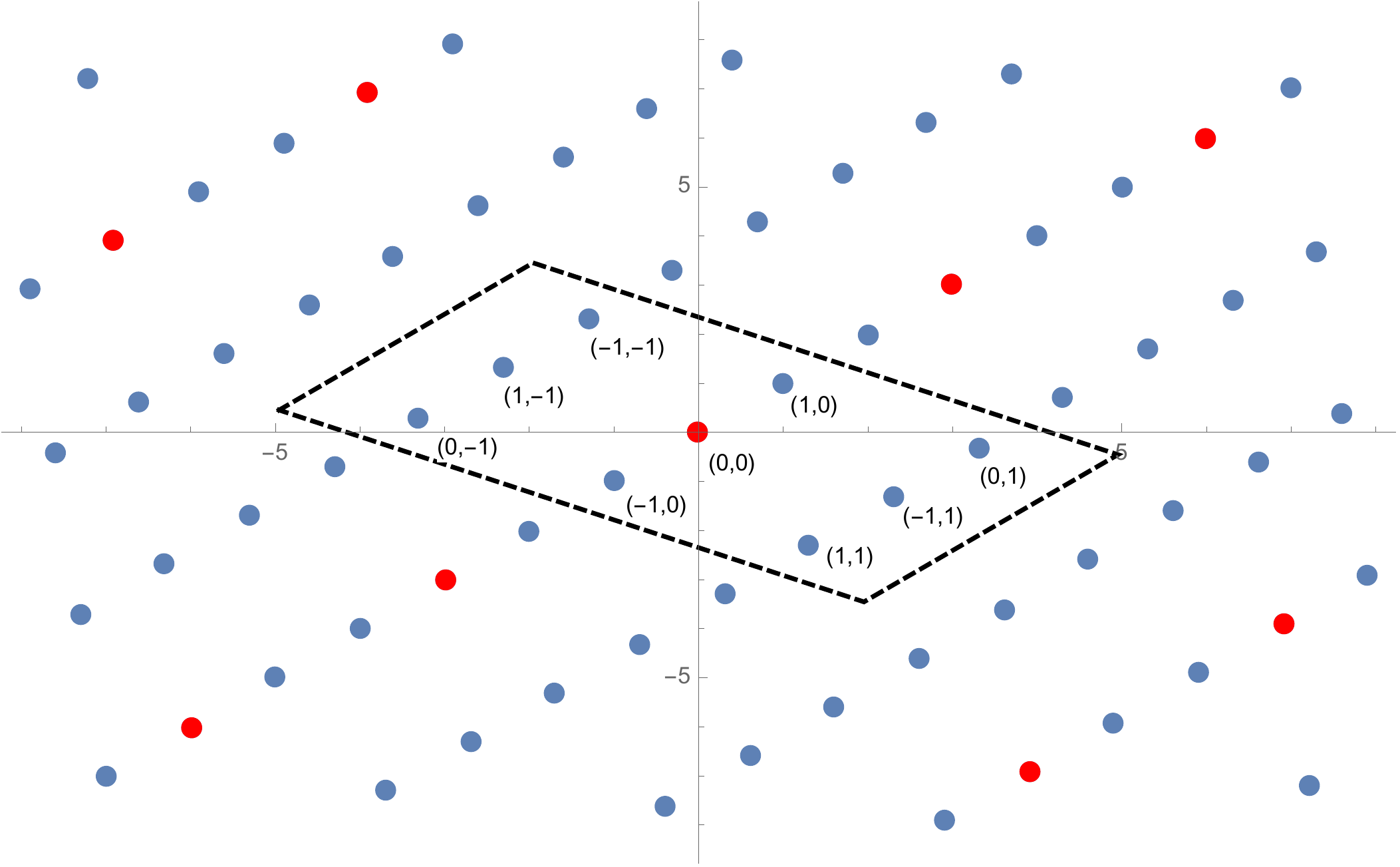}
		\label{fig:1}
		\caption{Illustration of the mapping $\rho$ (Eq. \eqref{eq:mappingrho}). Red points are elements of $3\mathbb{Z}[\mu]$.}
	\end{figure}
\end{exe}

Let $\beta > 0$ be a constant and $\alpha = ({\beta^{1/n}2^{-1} p^{1-k/n} \Delta_K^{1/2n}})^{-1}$ a normalization factor. Consider the ensemble of all lattices from generalized Construction A, normalized to volume $1/\beta$
\begin{equation}
\begin{split}
\mathbb{L}_{K, n,k,p,\beta} = \left\{ \alpha \Lambda_K(\mathcal{C}):\mathcal{C} \mbox{ is an }(n,k,p)\mbox{ code }\right\}.
\end{split}
\label{eq:genFading}
\end{equation}
Using the machinery developed by \cite{Loeliger}, and similarly to Appendix \ref{app:proofs} we can show that the ensemble is Minkowski-Hlawka for $p \to \infty$. The invariance of the lattices by units is described above. Let $*$ define the elementwise product between two vectors, i.e., $\xx * \yy = (x_1 y_1, \ldots, x_n y_n)$.

\begin{lemma}If $u \in \mathcal{O}_K^*$, then $\sigma(u) * \Lambda_{K}(\mathcal{C}_1) = \Lambda_{K}(\mathcal{C}_2)$, where $\mathcal{C}_1$ and $\mathcal{C}_2$ have the same dimension.
	\label{lem:equivalentLattices}
\end{lemma}
\begin{proof} It suffices to show that $\sigma(u) * \Lambda_{K}(\mathcal{C}_1) \subset \Lambda_{K}(\mathcal{C}_2)$, for some $\mathcal{C}_1$ and $\mathcal{C}_2$ of same rank, since both lattices have same volume due to the fact that $\sigma_1(u)\ldots \sigma_n(u) = \pm 1$. Let $\mathbf{y} = \sigma(u) * \lambda$, $\lambda = \sigma(x) \in\Lambda_{K}(\mathcal{C}_1)$. Then 
	$$\rho(\sigma(u)*\lambda) = (\phi \circ \pi)(\sigma(u)) * (\phi \circ \pi)(\sigma(x)) = \mathbf{a} * \mathbf{c},$$
	where $\mathbf{c}$ is a codeword. It follows that $\mathbf{y} \in \Lambda_K(\mathbf{a}*\mathcal{C}_1)$. Now, since no coordinate of $\mathbf{a}$ is zero, multiplication by $\mathbf{a}$ does not affect the rank of $\mathcal{C}_1$, which finishes the proof.
\end{proof}
\begin{rmk} In the previous lemma, $\mathcal{C}_1$ can be obtained from $\mathcal{C}_2$ by an equivalence of the Hamming metric.
\end{rmk}
From Lemma \ref{lem:equivalentLattices} the mapping $t_{\mathbf{u}} : \mathbb{L}_{K,n,k,p} \to \mathbb{L}_{K,n,k,p}$ given by $t_{\mathbf{u}}(\Lambda) = \sigma(u) * \Lambda$ is a bijection of the ensemble. This property is useful to handle deep fading.

\subsection{Infinite Lattice Constellations}
\label{sec:Infinite}
The dispersion and the Poltyrev limit of infinite constellations for the stationary ergodic fading channel was analyzed in \cite{DBLP:journals/corr/Vituri13}. In this case, $\left\{h_i\right\}$ is a random process (not necessarily iid) for which $\mu = E[\log | h |]$ exists and 

\begin{equation}
\lim_{n\to\infty}P\left(\left| \frac{1}{n} \sum_{i=1}^n \log |h_i| - \mu \right| > \varepsilon \right) = 0,
\end{equation}
for any positive $\varepsilon>0$.  Corollary 4.1 of \cite{DBLP:journals/corr/Vituri13} implies that the smallest possible VNR for a sequence of lattices to have vanishing error probability is
\begin{equation}
\gamma^* = e^{-2 \mu} 2\pi e.
\label{eq:optimalVNR}
\end{equation}
Let $P_e(\Lambda)$ be the probability of error of an infinite lattice scheme in a fading channel. In view of this result we can define fading-good lattices.

\begin{defi} A sequence $\Lambda_n$ of lattices with increasing dimension is \textit{good for the ergodic fading channel} if for all VNR $\gamma_{\Lambda}(\sigma) > e^{-2\mu} \pi e$, $P_e(\Lambda)\to0$ as $n \to \infty$.
\end{defi}

It was proven in \cite{DBLP:journals/corr/Vituri13}, for iid fading processes under some regularity conditions, that there exists a sequence of fading-good lattices with $P_e(\Lambda) = O(1/n^2)$. The proof only requires a Minkoswki-Hlawka ensemble, and hence an immediate corollary is that Generalized Construction A lattices as in \eqref{eq:genFading} are also fading good. We provide next a different approach that explores the algebraic structure and obtains an exponential decay of $P_e(\Lambda)$ with respect to $n$ in iid fading channels.

%



\begin{thm} There exists a sequence of ergodic fading-good lattices from Construction \eqref{eq:genFading}. 
	\label{thm:fading-good}
\end{thm}
\begin{proof} See Appendix \ref{app:C}.
\end{proof}
Under the additional assumption that the random process converges exponentially to its mean, we obtain the following direct corollary
\begin{cor}If for any sufficiently small $\varepsilon > 0$
	\begin{equation}
	\lim_{n\to\infty} -\frac{1}{n}\log P\left(\left| \frac{1}{n} \sum_{i=1}^n \log |h_i| - \mu \right| > \varepsilon \right) = A > 0,
	\end{equation}
	then the probability of error $P_e(\Lambda)$ in Theorem \ref{thm:fading-good} decays exponentially to zero.
\end{cor}
This is the case, for instance, of non-degenerate iid (following from the Cramer-Chernoff bound). For more general processes satisfying this hypothesis see e.g. \cite{MixingProcesses}. An explicit calculation for the Rayleigh fading process can be found in \cite[Eq. (6)]{VehkalahtiLuzzi2015}.

We close this section with a remark on the role of the group of units in the proof of Theorem \ref{thm:fading-good}. The function $g(x)$ (Eq. \eqref{eq:gx}) used in our version of the Minkowski-Hlawka theorem is bounded by definition. Intuitively, the group of units protects the channel from deep fadings. For general lattices not constructed for number fields this need not be true. A way to circumvent this problem \cite{DBLP:journals/corr/Vituri13} is to assume regularity conditions on the fading process which essentially guarantees a sufficient fast decay of the probability of deep fading. However, apart from questions of generality, this assumption degrades the probability of error to $O(1/n^2)$.

\subsection{Power-Constrained Model}
\label{sec:powerConstrained}
Similarly to \ref{sec:finalScheme}, we shape the constellation using the discrete Gaussian distribution. Given a coding lattice $\Lambda$, in the receiver side, given $\HH$, the estimate $\hat{\xx}$ that maximizes the a-posteriori probability is (cf. Section \eqref{sec:finalScheme}):
\begin{equation}
\hat{\xx} = \arg\min_{x\in\Lambda} \left\| \mathbf{F} \yy - \mathbf{R} \xx\right\|,
\label{eq:MAP}
\end{equation}
where $\mathbf{R}$ and $\mathbf{F}$ are diagonal with
\begin{equation*}
|{R}_{ii}|^2 = \rho h_i^2+1 \mbox{ and } {F}_{ii} = \frac{\rho h_i}{\sqrt{\rho |h_i|^2+1}} \mbox{,  } \rho \triangleq \frac{\sigma_s^2}{\sigma^2}.
\end{equation*}
In other words, MAP decoding is equivalent to lattice decoding with a scaling coefficient in each dimension. Consider the channel equation after scaling the received vector by $\mathbf{F}$:
\begin{equation}
\bar{\yy} = \mathbf{F} \yy  = \mathbf{R}\xx + \zz',
\end{equation}
where $\zz' = (\mathbf{F}\HH-\mathbf{R})\xx + \mathbf{F} \zz$ is the equivalent noise. The probability of error of lattice decoding for $\bar{\yy}$ is 
\begin{equation} P_e(\Lambda) = E_{\HH}\left[P(\zz' \notin \mathcal{V}_{\mathbf{R}\Lambda})\right].
\end{equation}
Since MAP decoding performs at least as well as the decoder in the proof of Theorem \ref{thm:fading-good}, the probability $P_e(\Lambda) \to 0$ if 
\begin{equation*}\gamma_{\Lambda}(\sigma^2) > e^{-E_h[\log(|r|)]} \pi e,
\end{equation*}
where $r = \rho h^2+1$. Moreover, if $\Lambda_n$ is a sequence of lattices with vanishing flatness factor (the existence of such a sequence can be guaranteed by the Minkowski-Hlawka theorem, as in \cite[Appendix III]{LLBS_12}) then the average power of the constellation $P \to \sigma_s^2$ and any rate
\begin{equation}
\begin{split}
R &=  \log(2\pi e \sigma_s^2) - \frac{2}{n} \log(V(\Lambda)) -\varepsilon \\&
= E_h[\log|r|] - \varepsilon = C-\varepsilon
\end{split}
\end{equation}
is achievable.

\section{Conclusion and Discussion}
In this paper we have presented algebraic lattice codes that achieve the compound capacity of the MIMO channel, and, in particular of the block-fading channel (i.e., when the channel realization $\HH$ is a diagonal matrix). This shows that lattices constructed from algebraic number theory can achieve not only limiting performance metrics, such as the DMT, but also the capacity of compound channels. Moreover, we have shown that algebraic lattices allow for a natural sub-optimal decoupled decoder, that handles the channel realization $\HH$ in a pre-process phase. The gap to capacity is characterized by the covering properties of unit-lattices of Number Fields (or, in the broader scope MIMO channel, Division Algebras). Finding algebraic structures whose unit-lattices are good coverings provides a design criterion for choosing the best lattice codes in this context.

The results in this paper are of an information-theoretic nature, and follow from generalizations of random arguments, such as the Minkowski-Hlawka theorem of the Geometry of Numbers. Practical multi-level schemes are the next natural steps for our constructions, and are part of ongoing work. Furthermore, the compound channel is a natural model to secrecy, since it is natural to suppose that no (or very few) previous knowledge of an eavesdropper channel can be available to a transmitter. A generalization of our methods to this model is currently under investigation.
\section{Acknowledgments}
The authors would like to thank Laura Luzzi, Roope Vehkalahti and Ling Liu for useful discussions and comments on previous versions of the manuscript. The first author acknowledges Sueli Costa for hosting him at the University of Campinas, where part of this work was developed.

\appendices
\section{Theorem \ref{thm:infiniteLimit}}
\label{app:proofs}
We show how to obtain bound \eqref{eq:boundCompact} in the proof of Theorem \ref{thm:infiniteLimit}. This is a special case of  \cite[Lem. 7]{RootVarayia1968}, included here for the sake of completeness. We first analyze the case $T = 1$. Let $\Omega_{\mathbf{\HH}}= \mathbf{E}_{\tilde{\HH}}^{\dagger}\mathbf{E}_{\tilde{H}}/\sigma_w^2$. The pdf of $\mathbf{E}_{\tilde{\HH}}^{-1}\tilde{\mathbf{W}}$ is
\begin{equation*}
f_{\mathbf{H}}(\tilde{\zz})= \frac{1}{\sqrt{\pi^n D}} e^{-\tilde{\zz}^{\dagger} \Omega_{\HH} \tilde{\zz}}.
\end{equation*}
Consider now a second matrix satisfying $\left\| \mathbf{E}_{\tilde{H}_0}- \mathbf{E}_{\tilde{H}}\right\| \leq \eta,$ which has corresponding pdf $f_{\mathbf{H}_0}(\tilde{z})$. The ratio between the two pdfs is
\begin{equation*}
\frac{f_{\mathbf{H}}(\tilde{\zz})}{f_{\mathbf{H}_0}(\tilde{\zz})}= e^{-\tilde{\zz}^{\dagger} \left(\Omega_{\HH}-\Omega_{\HH_0} \right)\tilde{\zz}}.
\end{equation*}
Now suppose that $\phi(\mathbf{z}) \in \mathcal{B}_{\sqrt{n(\sigma_w^2+\varepsilon)}}$.
\begin{equation*}
\begin{split} | \tilde{\mathbf{z}}^{\dagger} \left(\Omega_{\HH}-\Omega_{\HH_0} \right)\tilde{\mathbf{z}}| &\leq \left\|\tilde{\mathbf{z}}\right\|^{2} \left\|\Omega_{\HH}-\Omega_{\HH_0} \right\|_2^2 \leq {n(\sigma_w^2+\varepsilon)} \left\|\Omega_{\HH}-\Omega_{\HH_0} \right\|^2 \\ &= {n\left(1+\frac{\varepsilon}{\sigma_w^2}\right)} \left\|\mathbf{E}_{\tilde{\HH}_0}^{\dagger} \mathbf{E}_{\tilde{H}_0}- \mathbf{E}_{\tilde{H}}^{\dagger}  \mathbf{E}_{\tilde{\HH}}\right\| \\
&\leq {n\left(1+\frac{\varepsilon}{\sigma_w^2}\right)} \left(\left\|\mathbf{E}_{\tilde{\HH}_0}^{\dagger}\right\| \left\|\mathbf{E}_{\tilde{\HH}_0}- \mathbf{E}_{\tilde{\HH}}\right\|+\left\|\mathbf{E}_{\tilde{\HH}}^{}\right\| \left\|\mathbf{E}_{\tilde{\HH}_0}^{\dagger}- \mathbf{E}_{\tilde{\HH}}^{\dagger}\right\|\right)
\\
&\leq {n\left(1+\frac{\varepsilon}{\sigma_w^2}\right)} \left(2\alpha \eta \right).
\end{split}
\end{equation*}
Here $\left\|  \cdot \right\|_2$ denotes the operator norm of a matrix, which is upper bounded by the Frobenius norm. It follows that the ratio between the pdfs is bounded by

\begin{equation*}\frac{f_{\mathbf{H}}(\tilde{z})}{f_{\mathbf{H}_0}(\tilde{z})} \leq e^{\eta \alpha n\left(1+\frac{\varepsilon}{\sigma_w^2}\right)}.
\end{equation*}
For the case, $T>1$, replace the error matrices by the tensor product $\mathcal{E}_{\tilde{\HH}} = I_T \otimes \mathbf{E}_{\tilde{\HH}}$ and notice that 
$$\left\| \mathcal{E}_{\tilde{\HH}}^{\dagger}\mathcal{E}_{\tilde{\HH}} - \mathcal{E}_{\tilde{\HH}}^{\dagger}\mathcal{E}_{\tilde{\HH}} \right\|_2 = \left\|\mathbf{E}_{\tilde{\HH}_0}^{\dagger} \mathbf{E}_{\tilde{\HH}_0}- \mathbf{E}_{\tilde{\HH}}^{\dagger}  \mathbf{E}_{\tilde{\HH}}\right\|_2.$$

\section{Two Versions of The Minkoswki-Hlawka Theorem}
\subsection{Number Fields}
\label{app:numberFields}

In this appendix we prove the average behavior \eqref{eq:ensembleAverage}. For simplicity, we consider the case when the prime $p$ splits (i.e., $l = 1$, and the codes considered are from $\mathbb{F}_p$). We follow the proof of \cite{Loeliger}, recently adapted to Construction A over quadratic number fields in \cite{Adaptative}.

First, we show that the construction is well-defined for an infinite quantity of primes $p$. The following result is a consequence of from Chebotarev's Density Theorem. Let $\mathcal{P} \subset \mathbb{N}$ be the set of primes and define the Dirichlet density of a set $M \subset P$ (when the limit exists) as 
\begin{equation} \delta(M) = \lim_{s\to1^+}  \frac{\sum_{p\in M}{1/p^s}}{ \sum_{p\in P}{1/p^s}}.
\end{equation}
It follows that if the density is non-zero, than $M$ must be infinite.

\begin{thm}[\cite{Algebraic1},{Cor. 13.6 p. 547}]Let $K/\mathbb{Q}$ be an extension of degree $m$ and let $M$ be the set of primes which completely split in $K$. Then $\delta(M)=1/\left[ \overline{K} : \mathbb{Q} \right]$, where $\overline{K}$ is the Galois closure of $K$. In particular $M$ is infinite.
\end{thm}
For the case when $K/\mathbb{Q}$ is Galois, the density of primes which split is $\delta(M) = 1/n$. In particular, the theorem above implies that, for any given $K$, there exists an infinite set of primes for which the $\mathcal{O}_K$ Construction A is possible from codes over $\mathbb{F}_p$.

Let $f: \mathbb{R}^{2nT} \to \mathbb{R}$ be a Riemman integrable function with bounded support. Let $\Lambda = (1/\alpha)\Lambda_K(\mathcal{C})$ be a set in ensemble \eqref{eq:ensemble}, associated to the reduction $(\phi \circ \pi)$, as in Section \ref{sec:consA}. We first prove:

\begin{lemma} With $f, \phi, \pi$ and $\sigma$ as above, 
\begin{equation}\lim_{p \to \infty} \sum_{ {\mathbf{v}=\sigma(\xx): \above 0pt (\phi\circ\pi)(\xx)=\mathbf{0}}} f(\psi(\alpha^{-1} \mathbf{ v })) = 0.
\end{equation}
\label{lem:fixedLemma}
\end{lemma}
\begin{proof}
If $(\phi\circ\pi)(\mathbf{x})=\mathbf{0}$ then $\mathbf{x} \in \mathfrak{p}^T$, and for each component $x_j$, $N_{K/\mathbb{Q}}(x_j) = \Pi_{i=1}^n \sigma_i(x_j) \overline{\sigma}_i(x_j) \in p \mathbb{Z}$, which implies
$$\left\| \alpha^{-1} \mathbf{v} \right\|^2 = \alpha^{-2} \sum_{i=1}^n \sum_{j=1}^T |\sigma_i(x_j)|^2 \geq T \alpha^{-2} p^{1/n}.$$
Since $\alpha^{-2} p^{1/n} \to \infty$ as $p \to \infty$, and from the fact that $f$ has bounded support, $f(\alpha^{-1}\mathbf{v}) = 0$ for $p$ sufficiently large.
\end{proof}

Let $\mathcal{C}(T,k)$ be the set of all $(T,k)$ codes in $\mathbb{F}_{p}^T$. From Loeliger's averaging lemma \cite{Loeliger}, for a function $g:\mathbb{F}_{p}^T\to \mathbb{R}$:

\begin{equation}
\frac{1}{| \mathcal{C}(T,k) |} \sum_{\mathcal{C} \in \mathcal{C}(T,k)} \sum_{c \in \mathcal{C}\backslash \left\{ \mathbf{0} \right\}} g(c)= \frac{p^{k}-1}{p^{T}-1} \sum_{v \in \mathbb{F}_{p}^{T}\backslash \left\{ \mathbf{0} \right\}}g(v).
\end{equation}

\begin{thm}[Minkowski-Hlawka for the Generalized Construction A]
\begin{equation}
\begin{split}
\lim_{p\to\infty} \frac{1}{| \mathcal{C}(T,k) |} \sum_{\mathcal{C} \in \mathcal{C}(T,k)}\sum_{v \in \alpha^{-1} \Lambda_K(\mathcal{C}) \backslash\left\{\mathbf{0}\right\}} f(\psi(\mathbf{v})) = {\lambda^{-1}}\int_{\mathbb{R}^{2nT}} f(\mathbf{v}) d\mathbf{v}.
\end{split}
\label{eq:MHForGener}
\end{equation}
\end{thm}
\begin{proof}
To simplify the notation, let $| \mathcal{C}(T,k) |=M$, and $\mathcal{C} \backslash\left\{\mathbf{0}\right\} = \mathcal{C}^*$.

\begin{equation}
\begin{split}
&\frac{1}{M}\sum_{\mathcal{C} \in \mathcal{C}(T,k)}\sum_{v \in \Lambda_K(\mathcal{C}) \backslash\left\{\mathbf{0}\right\}} f(\alpha^{-1} v) = \frac{1}{M}\sum_{\mathcal{C} \in \mathcal{C}(T,k)}\sum_{{v=\sigma(x): \above 0pt (\phi\circ\pi)(x)=\mathbf{0}}} f(\alpha^{-1} v) + \frac{1}{M}\sum_{\mathcal{C} \in \mathcal{C}(T,k)}\sum_{{v=\sigma(x): \above 0pt (\phi\circ\pi)(x) \in \mathcal{C}^*}}f(\alpha^{-1} v)\\ & \stackrel{(a)}{=}\frac{1}{M}\sum_{\mathcal{C} \in \mathcal{C}(T,k)}\sum_{{v=\sigma(x): \above 0pt (\phi\circ\pi)(x)=\mathbf{0}}} f(\alpha^{-1} v)  +  \frac{p^{k}-1}{p^{T}-1}\sum_{{v=\sigma(x): \above 0pt (\phi\circ\pi)(x) \in \mathbb{F}_{p^l}^T \backslash\left\{\mathbf{0}\right\}}} f(\alpha^{-1} v). \\
&=\frac{1}{M}\sum_{\mathcal{C} \in \mathcal{C}(T,k)}\sum_{{v=\sigma(x): \above 0pt (\phi\circ\pi)(x)=\mathbf{0}}} f(\alpha^{-1} v)  +  \frac{p^{k}-1}{p^{T}-1}\sum_{{v\in\sigma(\mathcal{O}_K)\backslash\left\{\mathbf{0}\right\}}} f(\alpha^{-1} v). 
\end{split}
\end{equation}
Equation (a) is due to the averaging lemma. In the last equation, as $p \to \infty$, the first sum vanishes due to Lemma \ref{lem:fixedLemma} while the second one tends to the integral in rhs of \eqref{eq:MHForGener} (see, e.g., \cite[Thm. 2]{Loeliger})
\end{proof}

From \cite[Thm. 4]{Loeliger}, and the remark that follows it, we conclude that there is exists a family of AWGN good lattices from the ensemble of $\mathcal{O}_K$ lattices, for any $K$. Scaling the lattices appropriately, we get Eq. \eqref{eq:ensembleAverage}, and evaluating the probability of error as in \cite{Loeliger} we get the expression in parenthesis in Eq. \eqref{eq:averageProb}.

\subsection{Division Algebras}
Consider a lattice $\Lambda_p(\mathcal{C}) = \beta^{-1}(\mathcal{C})$ and the normalization factor 
$$\alpha = (|C|^{-1} p^{t^2m} (2^{-m}\gamma^{m(m-1)/2} \sqrt{\Delta_K^t})^m)^{1/m}.$$
\label{subsec:DivisionAppendix} Suppose that $f:\mathbb{R}^{t^2n}\to \mathbb{R}$ is a function with bounded support (by abuse of notation $f(X)$, when applied to a matrix in $\mathbb{R}^{t^2 \times n}$, is regarded as $f$ of the vectorized version of $X$).

\begin{lemma} Let $\alpha > 0$:
\begin{equation}\lim_{\alpha p \to \infty} \sum_{ {\mathbf{a} \in \Lambda_p(\mathcal{C}) \nozero \above 0pt \beta(\mathbf{a}) = 0}} f(\alpha \psi(\mathbf{a})) = 0.
\end{equation}
\label{lem:fixedLemma}
\end{lemma}
\begin{proof}
%
Let $\mathbf{a}=(a_1,\ldots,a_n)$ with $a_i = x_{0i}+e x_1\ldots + e^{n-1} x_{n-1,i}$. If $\beta(\mathbf{a})=0$, then $x_{ij} \in p(\mathcal{O}_k)$, therefore $\det \psi(a_i) = p^{t} k$. Using a "trace-det" inequality, the norm of the vectorized vector $\mathbf{a}$ satisfies $\mbox{tr}(\psi(a_i)^t\psi(a_i)) \geq t (\det \psi(a_i)^2)^{1/t} \geq t p^2$. Hence, as $\alpha p$ grows, since $f$ is bounded, the limit follows.
\end{proof}

Let $\mathcal{C}_T$ be a balanced set of codes in $\mathbb{F}_p^{n\times n}$. We have the following

\begin{thm}[Minkowski-Hlawka for the Generalized Construction A]
\begin{equation}
\begin{split}
\frac{1}{| \mathcal{C}_T |} \sum_{\mathcal{C} \in \mathcal{C}_t}\sum_{v \in \alpha \Lambda \backslash\left\{\mathbf{0}\right\}} f(v)\lessapprox {\det({\alpha \Lambda})^{-\frac{1}{2}}} \int_{\mathbb{R}^{t^2n}} f(\mathbf{v}) d\mathbf{v}.
\end{split}
\end{equation}
\end{thm}
\begin{proof}
\begin{equation}
\begin{split}
\frac{1}{| \mathcal{C}_T |}\sum_{\mathcal{C} \in \mathcal{C}_t}&\sum_{v \in \Lambda \backslash\left\{\mathbf{0}\right\}} f(\alpha v) = \frac{1}{| \mathcal{C}_T |}\sum_{\mathcal{C} \in \mathcal{C}_t}\sum_{ {\mathbf{a} \neq 0 \above 0pt \beta(\mathbf{a}) =0}} f(\alpha \psi(\mathbf{a}))+\frac{1}{| \mathcal{C}_T |}\sum_{\mathcal{C} \in \mathcal{C}_t} \sum_{ {\mathbf{a} \neq 0 \above 0pt \beta(\mathbf{a}) \in \mathcal{C} \neq 0}} f(\alpha \psi(\mathbf{a}))\\
&\stackrel{(a)}{\approx} \frac{1}{| \mathcal{C}_T |}\sum_{\mathcal{C} \in \mathcal{C}_t}\sum_{ {\mathbf{a} \neq 0 \above 0pt \beta(\mathbf{a}) \in \mathcal{C} \cap (M_t(\mathbb{F}_p)^n)^*}} f(\alpha \psi(\mathbf{a})) \stackrel{(b)}{\approx} p^{1-t^2n} \sum_{{\mathbf{a} \neq 0 \above 0pt \beta(\mathbf{a}) \in (M_t(\mathbb{F}_p)^n)^*}} f(\alpha \psi(\mathbf{a})) \\
&\leq p^{1-t^2n} \sum_{{\mathbf{a} \neq 0 \above 0pt \mathbf{a} \in \Lambda}} f(\alpha \psi(\mathbf{a})) 
\end{split}
\end{equation}
Equations (a) is a consequence of Lemma \ref{lem:fixedLemma} and (b) follows from from the averaging lemma. From now, the same arguments as \cite{Loeliger} prove that the integral approaches the sum in the rhs of (b).
\end{proof}
\section{Proof of Theorem \ref{thm:fading-good}}
\label{app:C}
Consider the received vector $\mathbf{y} = \HH \mathbf{x} + \mathbf{z}.$ Let $\Delta = |h_1\ldots h_n|^{1/n}$ and let $\mathcal{S}$ be a ball of radius $(\Delta/e^{\mu-\delta}) \sqrt{n(\sigma_w^2+\varepsilon})$, for $\delta$ and $\varepsilon$ sufficiently small. Consider a decoder that assigns $\hat{\mathbf{x}} = \tilde{\mathbf{x}}$ if $\mathbf{y}$ can be written in a unique way as $\mathbf{y} =\HH\tilde{\mathbf{x}} + \mathbf{z}$, with $\tilde{\mathbf{x}} \in {\Lambda}, \mathbf{z} \in \mathcal{S}$, and ``error'' otherwise (in Loeliger's terminology \cite{Loeliger} an ambiguity decoder). For a set $\mathcal{M}\subset \mathbb{R}^n$, let 
\begin{equation*} N_{\mathcal{M}}(\Lambda) \triangleq \left| (\Lambda\nozero) \cap \mathcal{M}\right|.
\end{equation*}We upper bound the probability of error as
\begin{equation*}P_e(\Lambda | \HH) \leq P(\mathbf{z} \notin \mathcal{S} | \HH) + P( N_{\mathbf{z}-\mathbf{S}} (\mathbf{H} \Lambda) \geq 1| \mathbf{z} \in \mathcal{S}, \HH),
\end{equation*}
and therefore 
\begin{equation}P_e(\Lambda) \leq E_{\HH}\left[P(\mathbf{z} \notin \mathcal{S})\right]+ E_{\HH}\left[P( N_{\mathbf{z}-\mathbf{S}} (\mathbf{H} \Lambda) \geq 1| \mathbf{z} \in \mathcal{S})\right].
\label{eq:Bound1}
\end{equation}
The first term does not depend on the chosen lattice and vanishes as $n \to \infty$. It can be bounded as
\begin{equation}\begin{split} P(\mathbf{z} \notin \mathcal{S}) &\leq P(\mathbf{z} \notin \mathcal{S} | \Delta > e^{\mu - \delta}) + P(\Delta < e^{\mu -\delta}) \\  &\leq P(\mathbf{z} \notin \mathcal{B}_{\sqrt{n(\sigma^2+\varepsilon)}}) + P(\Delta < e^{\mu -\delta}).
\end{split}
\end{equation}
The second term in the right-hand side of \eqref{eq:Bound1} can be upper bounded by
\begin{equation}
\begin{split}
& \int_{\mathbb{R}^n} \int_{\mathbb{R}^n} N_{\mathbf{z}-\mathbf{S}}(\mathbf{H} {\Lambda}) f_{\mathbf{z}�| \mathcal{S}}(z)f_{h}(\mathbf{h}) d\mathbf{h} d\mathbf{z}  .\\
\end{split}
\end{equation}
Let $\mathbb{L} = \mathbb{L}_{K,n,k,p}$ be the ensemble of Generalized Construction A lattices and consider the decomposition $\HH = \Delta \tilde{\EE}_{\tilde{\HH}} \mathbf{U}_{\tilde{\HH}}$, as in Definition \ref{def:compacify}. Let $f_{h}(\mathbf{h})$ be the pdf of the joint distribution of $(h_1,\ldots,h_n)$.
Taking the average over the ensemble (notice that at this point, for finite $p$, the ensemble is finite and we can commute integrals and sums):

\begin{equation*}
\begin{split}
& E_{\mathbb{L}}\left[\int_{\mathbb{R}^n} \int_{\mathbb{R}^n} N_{\mathbf{z}-\mathbf{S}}(\HH {\Lambda}) f_{\mathbf{z}�| \mathcal{S}}(z)f_{h}(\mathbf{h}) d\mathbf{h} d\mathbf{z}\right] =  \\
&\int_{\mathbb{R}^n} \int_{\mathbb{R}^n}  E_{\mathbb{L}}\left[N_{\mathbf{z}-\mathbf{S}}(\HH {\Lambda})\right] f_{\mathbf{z}�| \mathcal{S}}(z)f_{h}(\mathbf{h}) d\mathbf{h} d\mathbf{z} \stackrel{(a)}{=}\\
&\int_{\mathbb{R}^n} \int_{\mathbb{R}^n}  E_{\mathbb{L}}\left[N_{\mathbf{z}-\mathbf{S}}(\Delta \tilde{\EE}_{\tilde{\HH}} {\Lambda})\right] f_{\mathbf{z}�| \mathcal{S}}(z)f_{h}(\mathbf{h}) d\mathbf{h} d\mathbf{z} = \\
&E_{\mathbb{L}}\left[\sum_{\mathbf{x} \in \Lambda\nozero} \int_{\mathbb{R}^n} \int_{\mathbb{R}^n} \mathbbm{1}_{\Delta^{-1}\tilde{\EE}_{\tilde{\HH}}^{-1}(\mathbf{z}-\mathbf{S})}(\mathbf{x})f_{\mathbf{z}�| \mathcal{S}}(z)f_{h}(\mathbf{h}) d\mathbf{h} d\mathbf{z}\right] ,
\end{split}
\end{equation*}
where (a) is due to the fact that multiplication by unit is a bijection of the ensemble. Now take 
\begin{equation}g(\mathbf{x}) =\int_{\mathbb{R}^n} \int_{\mathbb{R}^n} \mathbbm{1}_{\Delta^{-1} \tilde{\EE}_{\tilde{\HH}}^{-1}(\mathbf{z}-\mathbf{S})}
(\mathbf{x})f_{\mathbf{z}�| \mathcal{S}}(z)f_{h}(\mathbf{h}) d\mathbf{h} d\mathbf{z}.
\label{eq:gx}
\end{equation}
We argue that $g(\mathbf{x})$ has bounded support. In effect, if $\mathbf{x}$ is such that $\left\| \mathbf{x} \right\| > 2 (C_n/e^{\mu-\delta})\sqrt{n(\sigma_w^2+\varepsilon)}$, then $\left\| \Delta \tilde{\EE}_{\tilde{\HH}} \mathbf{x} \right\| > 2 (\Delta/e^{\mu-\delta}) \sqrt{n(\sigma_w^2+\epsilon)}$, which implies that $\mathbbm{1}_{\Delta^{-1}\tilde{\EE}_{\tilde{\HH}}^{-1}(z-S)}(\mathbf{x})=0$ and $g(\mathbf{x}) = 0$, therefore

\begin{equation}
\lim_{p\to \infty} E_{\mathbb{L}}\left[\sum_{\mathbf{x} \in \Lambda \nozero} g(\mathbf{x}) \right] =  {{ e^{-n(\mu - \delta)}}\beta \,\, \vol\,{\mathcal{B}_{\sqrt{n(\sigma_w^2+\varepsilon)}}}}.
\end{equation}
Therefore, there exists a sequence of lattices in the random ensemble such that $P_e(\Lambda)$ decays to zero, as long as the VNR is greater $\gamma^*$ (Eq. \eqref{eq:optimalVNR}).

\bibliographystyle{unsrt}
\bibliography{block_fading}

\end{document}